\documentclass{article}
\usepackage[utf8]{inputenc}
\usepackage[dvipsnames]{xcolor}
\usepackage[english]{babel}
\usepackage{graphicx}
\usepackage{wrapfig}
\usepackage{caption}
\usepackage{subcaption}
\usepackage{rotating}
\usepackage{float}  
\usepackage{longtable}
\graphicspath{ {figures/} }
\usepackage{array}
\usepackage{lscape}
\usepackage{epstopdf}
\usepackage{amsthm}
\usepackage{amsmath}
\usepackage{enumerate}
\usepackage{amssymb}
\usepackage{tabularx}
\usepackage{authblk}
\usepackage{a4wide}
\usepackage[hidelinks]{hyperref}
\makeatletter

\renewcommand*\env@matrix[1][*\c@MaxMatrixCols c]{%
  \hskip -\arraycolsep
  \let\@ifnextchar\new@ifnextchar
  \array{#1}}
\makeatother
\usepackage{systeme}

\usepackage[nottoc]{tocbibind}
\usepackage{csquotes}

\hypersetup{
	bookmarks=true,         
	unicode=false,          
	pdftoolbar=true,        
	pdfmenubar=true,        
	pdffitwindow=false,     
	pdfstartview={FitH},    
	pdftitle={My title},    
	pdfauthor={Author},     
	pdfsubject={Subject},   
	pdfcreator={Creator},   
	pdfproducer={Producer}, 
	pdfkeywords={keywords}, 
	pdfnewwindow=true,      
	colorlinks=true,       
	linkcolor=blue,          
	citecolor=blue,        
	filecolor=magenta,      
	urlcolor=cyan           
}

\newtheorem{theorem}{Theorem}
\newtheorem{lemma}{Lemma}

\newtheorem{proposition}{Proposition}

\newtheorem{algorithm}{Algorithm}
\newtheorem*{remark}{Remark}

\title{Cost optimisation of hybrid institutional incentives for promoting cooperation in finite populations}
\author[$\dagger$]{M. H. Duong}
\author[$\star$]{C. M. Durbac}
\author[$\ddag$]{T. A. Han}

\affil[$\dagger$]{School of Mathematics, University of Birmingham, UK. Email: h.duong@bham.ac.uk}
\affil[$\star$]{School of Mathematics, University of Birmingham, UK. Email: c.m.durbac@bham.ac.uk}
\affil[$\ddag$]{ School of Computing, Engineering and Digital Technologies, Teesside University, UK. Email: T.Han@tees.ac.uk}
\date\today

\begin{document}

\maketitle
\begin{abstract}
In this paper, we rigorously study the problem of cost optimisation of hybrid (mixed) institutional incentives, which are a plan of actions involving the use of reward and punishment by an external decision-maker, for maximising the level (or guaranteeing at least a certain level) of cooperative behaviour in a well-mixed, finite population of self-regarding individuals who interact via cooperation dilemmas (Donation Game or Public Goods Game). We show that a mixed incentive scheme can offer a more cost-efficient approach for providing incentives while ensuring the same level or standard of cooperation in the long-run. We establish the asymptotic behaviour (namely neutral drift, strong selection, and infinite-population limits). We prove the existence of a phase transition, obtaining the critical threshold of the strength of selection at which the monotonicity of the cost function changes and providing an algorithm for finding the optimal value of the individual incentive cost. Our analytical results are illustrated with numerical investigations.
Overall, our analysis provides novel theoretical insights into the design of cost-efficient institutional incentive mechanisms for promoting the evolution of cooperation in stochastic systems. 
\end{abstract}
\newpage
\tableofcontents
\section{Introduction}
\paragraph{Literature review.} Evolutionary Game Theory made its debut in 1973 with John Maynard Smith’s and George R. Price’s work on the formalisation of animal contests, thus successfully using Classical Game Theory to create a new framework that could predict the evolutionary outcomes of the interaction between competing individuals. Ever since, it has been widely used to study myriad questions in various disciplines like Evolutionary Biology, Ecology, Physics, Sociology and Computer Science, including what the mechanisms underlying the emergence and stability of cooperation are \cite{nowak2006} and how to mitigate climate change and its risks \cite{riskSantos}.

Cooperation is the act of paying a cost  in order to convey a benefit to somebody else. Although it initially seems against the Darwinian theory of natural selection, cooperation has been, is, and will be a vital part of life, from cellular clusters to bees to humans \cite{sigmund2010calculus,nowak2011supercooperators}. Several mechanisms for promoting the evolution of cooperation have been identified, including kin selection, direct reciprocity, indirect reciprocity, network reciprocity, group selection and different forms of  incentives \cite{nowak2006,sigmund2010calculus,perc2017statistical,rand2013human,van2014reward}. The current work focuses on institutional incentives \cite{sasaki2012take,sigmundinstitutions,wang2019exploring,duong2021cost,cimpeanu2021cost,sun2021combination,van2014reward,gurerk},  which are a plan of actions involving the use of reward (i.e., increasing the payoff of cooperators) and punishment (i.e., reducing the payoff of defectors) by an external decision-maker, in particular, how they can be used in combination (i.e. hybrid or mixed incentives) in a cost-efficient way for maximising the levels of  cooperative behaviour in a well-mixed, finite population of self-regarding individuals. 
 
 Promoting and implementing cooperation via  incentives is costly to the incentive-providing institution, such as the United Nations or the European Union \cite{ostrom2009understanding,van2014reward}. Thus, it is crucial to understand how to  minimise the said cost while ensuring a desirable standard of cooperation. 
 In this work, players interact via cooperation dilemmas, both the pairwise Donation Game (DG) and its multi-player version, the  Public Goods Game (PGG) \cite{sigmund2010calculus,nowak2006}. These games have been widely adopted to model social dilemma situations in which collective rationality leads to individual irrationality. 

Several theoretical models studied how to combine institutional reward and punishment for enhancing the emergence and stability of cooperation \cite{chen2014optimal,gois2019reward,berenji2014recidivism,hilbe2010incentives,sun2021combination,gurerk}. However, little attention  has been given to addressing the cost optimisation of providing incentives. Chen et al.~\cite{carrotstick} looked at a rewarding policy that switches the incentive from reward to punishment when the frequency of cooperators exceeds a certain threshold. This policy establishes cooperation at a lower cost and under a wider range of conditions than either reward or punishment alone, in both well-mixed and spatial populations. Others have applied the `carrot and stick' idea to criminal recidivism and rehabilitation as now the justice system is switching its focus to reintegrating wrongdoers into society after the penalty has been served. Berenji et al.'s work \cite{berenji2014recidivism} studied the game where players decide to permanently reform or continue engaging in criminal activity, eventually reaching a state where they are considered incorrigible. Since resources may be limited, they fixed the combined rehabilitation and penalty costs per crime. The most successful strategy in reducing crime is to optimally allocate resources so that after being having served the penalty, criminals are reintroduced into society via impactful programs. Wang et al.~\cite{wang2019exploring} explored the optimal incentive that not only minimises the total cost, but also guarantees a sufficient level of cooperation in an infinite and well-mixed population via optimal control theory. 

This work however does not take into account various stochastic effects of evolutionary dynamics such as mutation and non-deterministic behavioural update. In a deterministic system of cooperators and defectors, once the latter disappear, there is no further change in the system and hence no further interference is needed. When mutation is present, defectors can reappear and become numerous over time, requiring the external institution to spend more on providing further incentives. Moreover, the intensity of selection - how strongly an individual bases their decision to copy another individual’s strategy on their fitness difference - is overlooked. When selection is weak, providing incentives would make little difference in causing behavioural change as no individual would be motivated to copy another and all changes in strategy would be due to noise. When selection is strong, incentives that ensure a minimum fitness advantage to cooperators would ensure a positive behavioural change as the players would be more likely to copy one another. Recently, in \cite{Zhang2022} by simulating the weak prisoner’s dilemma in finite populations, the authors find that a combination of appropriate punishment and reward mechanisms can promote cooperation's prosperity regardless of how large or small the temptation to defect is.

Whilst the above works were concerned with well-mixed populations, the following two studies deal with spatial ones. Szolnoki and Perc \cite{corposneg} looked at whether there are evolutionary advantages in correlating positive and negative reciprocity, as opposed to adopting only reward or punishment. Their work supports others that use empirical data. In those studies, the data fails to support the central assumption of the strong reciprocity model that negative and positive reciprocity are correlated. In a different work \cite{2ndfreeriding}, the authors showed how second-order free-riding on antisocial punishment restores the effectiveness of prosocial punishment, providing an evolutionary escape from adverse effects of antisocial punishment. Both these works use Monte Carlo simulations.

Moreover, several works have provided insights into how best to promote the emergence of collective behaviours such as cooperation and fairness while also considering the institutional costs of providing incentives  \cite{liu2018evolutionary,chen2014optimal, duong2021cost, cimpeanu2021cost,cimpeanu2019exogenous,han2018cost,cimpeanu2023does}, see also \cite{wang2021} for a recent survey on these papers.
Indeed, most relevant to our work are Duong and Han \cite{duong2021cost} and Han and Tran-Thanh \cite{han2018cost}, which derived analytical conditions for which a general incentive scheme can guarantee a given level of cooperation while at the same time minimising the total cost of investment. These results are highly sensitive to the intensity of selection. They also studied a class of incentive strategies that make an investment whenever the number of players with a desired behaviour reaches a certain threshold $t\in\{1,\ldots,N-1\}$ ($N$ is the population size), showing that there is a wide range of values for the threshold that outperforms standard institutional incentive strategies - those which invest in all players, i.e. the threshold is $t=N-1$ \cite{carrotstick}. These works however did not study the cost-efficiency of the mixed incentive scheme, which is the focus of the present work.

\paragraph{Overview of contribution of this paper.}  As mentioned above, in this work, we consider a well-mixed, finite population of self-regarding individuals where  players interact via cooperation dilemmas (DG and PGG) and rigorously study the problem of cost optimisation of hybrid institutional incentives (combination of reward and punishment) for maximising the levels of cooperative behaviour (or guaranteeing at least a certain levels of cooperation). This problem is challenging due to the number of parameters involved such as the number of individuals in the population, the strength of selection, the game-specific quantities, as well as the efficiency ratios of providing the corresponding incentive. In particular, the Markov chain modelling the evolutionary process is of order equal to the population size, which is large but finite. The calculation of the entries of the corresponding fundamental matrix is intricate, both analytically and computationally.

Our present work provides a rigorous and delicate analysis of this problem, combining techniques from different branches of Mathematics including Markov chain theory, polynomial theory, and numerical analysis. The main results of the paper can be summarised as follows (detailed and precise statements are provided in Section \ref{sec: models}).
\begin{enumerate}[(i)]
    \item We show that a mixed incentive scheme can offer a more cost-efficient approach for providing incentives while ensuring the same level or standard of cooperation in the long-run. 
    \item We obtain the asymptotic behaviour of the cost function in the limits of neutral drift, strong selection as well as infinite population sizes. 
    \item We prove the existence of a phase transition for the cost function, obtaining the critical threshold of the strength of selection at which the monotonicity of the cost function changes and finding the optimal value of the individual incentive cost. 
\end{enumerate}
Furthermore, we provide numerical simulations to illustrate the analytical results.
\subsection*{Organisation of the paper} The rest of the paper is organised as follows. In Section \ref{sec: models} we present the model, methods, and the main results. Our main results include Proposition \ref{prop: 1st result} on the efficiency of the combined reward and punishment incentives compared to implementing them separately, Theorem \ref{thm: main thm1} on the asymptotic behaviour (neutral drift, strong selection, and infinite population limits) of the cost function, and Theorem \ref{thm: main thm2} on the optimal incentive. In Section \ref{sec: reward and punishment} we provide detailed computations and proofs of Theorem \ref{thm: main thm2}. Proof of Theorem \ref{thm: main thm1} is given in Section \ref{sec: asymptotic behaviour}. Summary and further discussions are provided in Section \ref{sec: discussion}. Finally, Section \ref{sec: appendix} contains the proof of Proposition \ref{prop: 1st result}, detailed computations, and proofs for the technical results.

\section{Model, methods, and main results} 
\label{sec: models}
In this section, we present the model, methods, and main results of the paper. We first introduce the class of games, namely cooperation dilemmas, that we are interested in throughout this paper.

\subsection{Cooperation dilemmas} 

\noindent We consider a well-mixed, finite population of $N$ self-regarding individuals (players) who engage with one another using one of the following one-shot (i.e. non-repeated) cooperation dilemmas, the Donation Game (DG) or its multi-player version, the Public Goods Game (PGG). Strategy wise, each player can choose to either cooperate (C) or  defect (D). 

Let $\Pi_C(j)$ be the average payoff of a C player (cooperator) and $\Pi_D(j)$ that  of a D player (defector), in a population with $j$ $C$ players and $(N-j)$ $D$ players. As can be seen below, the difference in payoffs $\delta = \Pi_C(j) - \Pi_D(j)$ in both games does not depend on $j$. For the two cooperation dilemmas considered in this paper, namely the Donation Games and the Public Goods Games,  it is  always the case that $\delta < 0$. This does not cover some weak social dilemmas such as the Snowdrift Game, where $\delta>0$ for some $j$, the general prisoners' dilemma, and the collective risk game \cite{sun2021combination}, where $\delta$ depends on $j$. We will investigate these games in future research (see Section \ref{sec: discussion} for further discussion).
\subsubsection*{Donation Game (DG)} 
\noindent The Donation Game is a form of Prisoners' Dilemma in which cooperation corresponds to offering the other player a benefit $B$ at a personal cost $c$, satisfying that $B > c$. Defection means offering nothing. 
The payoff matrix of DG (for the row player) is given as follows  
\[
 \bordermatrix{~ & C & D\cr
                  C & B-c & -c \cr
                  D & B & 0  \cr
                 }. 
\]

\noindent Denoting $\pi_{X,Y}$ the payoff of a strategist  $X$ when playing with a strategist $Y$ from the  payoff matrix above, we obtain
\begin{equation*} 
\begin{split} 
\Pi_C(j) &=\frac{(j-1)\pi_{C,C} + (N-j)\pi_{C,D}}{N-1} = \frac{(j-1) (B-c) + (N-j) (-c)}{N-1}  ,\\
\Pi_D(j) &=\frac{j\pi_{D,C} + (N-j-1)\pi_{D,D}}{N-1} =\frac{j B}{N-1}.
\end{split}
\end{equation*} 
Thus, 
$$\delta = \Pi_C(j) - \Pi_D(j) =  -\Big(c + \frac{B}{N-1}\Big).$$

\subsubsection*{Public Goods Game (PGG)} 

\noindent In a Public Goods Game, players interact in a group of size $n$, where they decide to cooperate, contributing an amount $c > 0$ to a common pool, or to defect, contributing nothing to the pool. The total contribution in a group is multiplied by a factor $r$, where $1 < r < n$ (for the PGG to be a social dilemma), which is then shared equally among all members of the group, regardless of their strategy. Intuitively, contributing nothing offers one a higher amount of money after redistribution.

The average payoffs, $\Pi_C(j)$ and $\Pi_D(j)$, are calculated based on the assumption that the groups engaging in a public goods game are given by multivariate
hypergeometric sampling. Thereby, for transitions between two pure states, this reduces to sampling, without replacement, from a hypergeometric distribution. More precisely, we obtain \cite{hauert2007}
\begin{equation*} 
\begin{split} 
\Pi_C(j) &= \sum^{n-1}_{i=0}\frac{\dbinom{j-1}{i}\dbinom{N-j}{n-1-i}}{
 \dbinom{N-1}{n-1}} \ \left(\frac{(i+1)rc}{n} - c\right) =\frac{rc}{n}\left(1 + (j-1)\frac{n-1}{N-1}\right) - c ,\\
\Pi_D(j) &=\sum^{n-1}_{i=0}\frac{\dbinom{j}{i}\dbinom{N-1-j}{n-1-i}}{
 \dbinom{N-1}{n-1}} \ \frac{jrc}{n} =\frac{rc(n-1)}{n(N-1)}j.
\end{split}
\end{equation*} 
Thus, 
$$\delta = \Pi_C(j) - \Pi_D(j) = -c \left(1 - \frac{r(N-n)}{n(N-1)} \right).
$$

\subsection{Cost of institutional reward and punishment} 
To reward a cooperator (respectively, punish a defector), the institution has to pay an amount $\theta/a$ (resp., $\theta/b$) so that the cooperator's (defector's) payoff increases (decreases) by $\theta$, where $a, b > 0$ are constants representing the efficiency ratios of providing the corresponding incentive. 

In an institutional enforcement setting, we assume that the institution has full information about the population composition or statistics at the time of decision-making. That is, given the well-mixed population setting, we assume that  the number $j$ of cooperators in the population is known. Thus, if both reward and punishment are feasible options (i.e., mixed incentives), the institution can minimise its cost by choosing the minimum of $\frac{j}{a}$ and $\frac{N-j}{b}$. Thus, the key question here is: \textit{what is the optimal value of the individual incentive cost $\theta$ that ensures a sufficient desired level of cooperation in the population (in the long-run) while minimising the total cost spent by the institution?  }

Note that, as discussed above, this mixed incentive, also known as the `carrot and stick' approach, has been shown efficient for promoting cooperation in both pairwise and multi-player interactions \cite{carrotstick,hilbe2010incentives,sun2021combination,gois2019reward,gurerk}. However, these works have not studied cost optimisation and have not shown whether this approach is actually more cost-efficient and by how much.  

\subsubsection*{Deriving the expected cost of providing  institutional incentives}
\label{sec: cost of incentive}
In this model, we adopt the finite population dynamics with the Fermi strategy update rule \cite{traulsen2006}, stating that a player $X$ with fitness $f_X$ adopts the strategy of another player $Y$ with fitness $f_Y$ with a probability given by $P_{X,Y}=\left(1 + e^{-\beta(f_Y-f_X)}\right)^{-1}$, where $\beta$ represents the intensity of selection. 
We compute the expected number of times the population contains $j$ C players, $1 \leq j \leq N-1$. For that, we consider an absorbing Markov chain of $(N+1)$ states, $\{S_0, ..., S_N\}$, where $S_j$ represents a population with $j$ C players.  $S_0$ and $S_N$ are absorbing states. Let  $U = \{u_{ij}\}_{i,j = 1}^{N-1}$ denote the transition matrix between the $N-1$ transient states, $\{S_1, ..., S_{N-1}\}$. The transition probabilities can be defined as follows, for $1\leq i \leq N-1$: 
\begin{equation} 
\label{eq: transition probabilities}
\begin{split} 
u_{i,i\pm k} &= 0 \qquad \text{ for all } k \geq 2, \\
u_{i,i\pm1} &= \frac{N-i}{N} \frac{i}{N} \left(1 + e^{\mp\beta[\Pi_C(i) - \Pi_D(i)+\theta]}\right)^{-1},\\
u_{i,i} &= 1 - u_{i,i+1} -u_{i,i-1}.
\end{split} 
\end{equation}
The entries $n_{ik}$ of the so-called fundamental matrix $\mathcal{N}=(n_{ik})_{i,k=1}^{N-1}= (I-U)^{-1}$ of the absorbing Markov chain gives the expected number of times the population is in the state $S_j$ if it is started in the transient state $S_i$ \cite{kemeny1976finite}.
As a mutant can  randomly occur either at $S_0$ or $S_N$, the expected number of visits at state $S_i$ is thus, $\frac{1}{2} (n_{1i} + n_{N-1,i})$.
The total cost per generation is
\begin{equation*}
\theta_j = \min\Big(\frac{j}{a}, \frac{N-j}{b}\Big) \times \theta.
\end{equation*}

Hence, the expected total cost of interference for mixed reward and punishment is 
\begin{equation} 
\label{eq:total_investment}
E_{mix}(\theta) = 
\frac{\theta}{2} \sum_{j=1}^{N-1}(n_{1j} + n_{N-1,j}) \min\Big(\frac{j}{a}, \frac{N-j}{b}\Big).
\end{equation}
As a comparison, we recall the cost for reward and punishment incentives, $E_r$ and $E_p$,  respectively, when being used separately \cite{duong2021cost}  
\begin{equation}
\label{eq: Er and Ep}    
E_{r}(\theta) = \frac{\theta}{2}\sum_{j=1}^{N-1}(n_{1j}+n_{N-1,j})\frac{j }{a},\quad E_{p}(\theta)=\frac{\theta}{2}\sum_{j=1}^{N-1}(n_{1j}+n_{N-1,j})\frac{N-j}{b}.
\end{equation}

By comparing \eqref{eq:total_investment} and \eqref{eq: Er and Ep} one clearly expects that the efficiency ratios $a$ and $b$ strongly affect the incentive cost. In the cost functions $E_r$ and $E_p$, they are just scaling factors and do not affect the analysis of these functions. This is not the case in the combined incentive. One of the main objectives of this paper is to reveal explicitly the influence of $a$ and $b$ on the cost function. From a mathematical point of view, the presence and interplay of $a$ and $b$ make the analysis of the combined incentive much harder than that of the separate ones.

\begin{remark}[On the interference scheme]
In the mixed incentive setting being considered in this paper, the institution either rewards every cooperator or punishes every defector, depending on which one is less costly. Although being rather unsophisticated, this incentive strategy is typically considered in the literature of institutional incentives modelling. However, other interference schemes are also investigated in many works, for instance, the institution only rewards C players whenever their frequency or number in the population does not exceed a given threshold $t$, where $1\leq t\leq N-1$. The scheme studied in this paper corresponds to the case where $t=N-1$. We refer the reader to \cite{han2018cost} and references therein for more information about different interference schemes. We plan to generalise the results of this paper to more complicated incentive strategies in future work, see Section \ref{sec: discussion} for further discussion. 
\end{remark}

\subsubsection*{Cooperation frequency}
Since the population consists of only two strategies, the fixation  probabilities of a C (D) player in a homogeneous population of D (C) players  when the interference scheme is carried out are, respectively, \cite{nowak}
\begin{equation*} 
\begin{split}
\rho_{D,C} &= \left(1+\sum_{i = 1}^{N-1} \prod_{k = 1}^i \frac{1+e^{\beta(\Pi_C(k)-\Pi_D(k) + \theta)}}{1+e^{-\beta(\Pi_C(k)-\Pi_D(k)+\theta)}}  \right)^{-1}, \\
\rho_{C,D} &= \left(1+\sum_{i = 1}^{N-1} \prod_{k = 1}^i \frac{1+e^{\beta(\Pi_D(k)-\Pi_C(k) - \theta)}}{1+e^{-\beta(\Pi_D(k)-\Pi_C(k)-\theta)}}  \right)^{-1}.
\end{split}
\end{equation*} 

Computing the stationary distribution using  these fixation probabilities, we  obtain the frequency of cooperation  $$\frac{\rho_{D,C}}{\rho_{D,C}+\rho_{C,D}}.$$

Hence, this frequency of cooperation can be maximised by maximising 
\begin{equation}
\label{eq:max}
\max_{\theta} \left(\rho_{D,C}/\rho_{C,D}\right).  
\end{equation} 

The fraction in Equation~\eqref{eq:max} can be simplified as follows \cite{nowak2006} 
\begin{eqnarray}
\nonumber
\frac{\rho_{D,C}}{\rho_{C,D}} &=&  \prod_{k = 1}^{N-1} \frac{u_{i,i-1}}{u_{i,i+1}} =\prod_{k = 1}^{N-1} \frac{1 + e^{\beta[\Pi_C(k)-\Pi_D(k) + \theta]}}{1 + e^{-\beta[\Pi_C(k)-\Pi_D(k) + \theta]}} \\
\nonumber
&=& e^{\beta\sum_{k = 1}^{N-1} \left(\Pi_C(k)-\Pi_D(k) + \theta\right)} \\
\label{eq:max_Q_prime}
 &=& e^{\beta (N-1)(\delta +  \theta)}. 
 \end{eqnarray}

In the above transformation, $u_{i,i-1}$ and $u_{i,i-1}$ are the probabilities  to decrease or increase the number  of C players  (i.e. $i$) by one in each time step, respectively. 

We consider non-neutral selection, i.e.  $\beta > 0$ (under neutral selection, there is no need to use incentives as no player is likely to copy another player and any changes in strategy that happen are due to noise as opposed to payoffs). Assuming that we desire to obtain  at least an $\omega  \in [0,1]$ fraction of cooperation, i.e. $\frac{\rho_{D,C}}{\rho_{D,C}+\rho_{C,D}} \geq \omega$, it follows from equation~\eqref{eq:max_Q_prime}  that
\begin{equation} 
\label{eq:omega_fraction}
 \theta \geq \theta_0(\omega) = \frac{1}{(N-1)\beta} \log\left(\frac{\omega}{1-\omega}\right) - \delta .
\end{equation}

Therefore it is guaranteed that if $\theta  \geq \theta_0(\omega)$, at least an $\omega$ fraction of cooperation can be expected.
This condition implies that the lower bound of $\theta$ monotonically depends on $\beta$. Namely, when $\omega \geq 0.5$, it increases with $\beta$ and when $\omega < 0.5$, it decreases with $\beta$.\\

To summarise, we obtain the following constrained minimisation problem
\begin{equation}
\label{eq: min prob} \min_{\theta\geq \theta_0} E_{mix}(\theta).
\end{equation}
\begin{remark}[On the formula of the cost  of the mixed incentive]
In the derivation of the cost function \eqref{eq:total_investment}, we  assumed that the population is equally likely to start in the homogeneous state $S_0$ as well as in the homogeneous state $S_N$. However, in general, this might not be correct. For example, if cooperators are very likely to fixate in a population of defectors, but defectors are unlikely to fixate in a population of cooperators, mutants are on average more likely to appear in the homogeneous cooperative population (that is in $S_N$). Similarly, the population might also be likely to appear in $S_0$ rather than $S_N$. In general, in the long-run, the population will start at $i = 0$ ($i = N$, respectively) with probability equal to the frequency of D (C) computed at the equilibrium, $f_D = 1/(r+1)$ ($f_C = r/(r+1)$, respectively), where $r = e^{\beta (N-1)(\delta +  \theta)}$. Thus generally, the expected number of visits at state $S_i$ will be $ f_D n_{1i} + f_C n_{N-1,i}$. Therefore, instead of \eqref{eq:total_investment}, in the general setting the formula for the cost function should be
$$
E_{mix}= \sum_{j=1}^{N-1}(f_D n_{1j} + f_C n_{N-1,j}) \min\Big(\frac{j}{a}, \frac{N-j}{b}\Big).
$$
In practice, in many works based on agent-based simulations \cite{chen2014optimal,cimpeanu2023does,2ndfreeriding,han2018fostering,carrotstick}, it is often assumed that mutation is negligible and simulations end whenever the population fixates in a homogeneous state. Moreover, these simulations usually assume that the initial population starts at a homogeneous state or has a uniform distribution of different types. In this work, we thus assume an equal likelihood that the population starts at one of the homogeneous states and our formula \eqref{eq:total_investment} captures such scenarios.  This assumption enables us to analytically study the cost function and its behaviour. As will be clear in the subsequent sections, the analysis is already very complicated in this simplified setting. Our results encapsulate the  intermediate-run dynamics, an approximation that is valid if the time-scale is long enough for one type to reach fixation, but too short for the next mutant to appear. Our findings might thus be more practically useful for the optimisation of the institutional budget for providing incentives on an intermediate timescale. 

We will study this problem in the most general case, where the initial population can start at an arbitrary state, in future work. The cost function can be obtained from extensive agent-based simulations of the evolutionary process. However, this approach is very computationally expensive especially when one wants to analytically study the cost function as a function of the individual incentive cost for large population sizes, which is the focus of this paper. Previous works have already shown that outcomes from our adopted evolutionary processes (small-mutation limit) are in line with extensive agent-based simulations, e.g. in \cite{han2013emergence,van2012emergence,hauert2007,sigmundinstitutions}.
\end{remark}

\subsection{Main results of the present paper}
\begin{proposition}[Combined incentives vs separate ones]
\label{prop: 1st result}
It is always more cost efficient  to use the mixed incentive approach than  a separate incentive, reward or punishment, 
$$
E_{mix}\leq \min\{E_r, E_p\}.
$$
If $\frac{b}{a}\leq \frac{1}{N-1}$, then $E_{mix}(\theta)=E_r(\theta)$.  If $\frac{b}{a}\geq N-1$, then $E_{mix}(\theta)=E_p(\theta)$.
That is, if providing reward for a cooperator is much more cost-efficient for the institution than punishing a defector, i.e., when $b/a \geq N-1$, then it is optimal to use reward entirely. Symmetrically, if punishment is much more efficient,  i.e. $a/b \geq N-1$, then it is optimal to use punishment entirely. Otherwise, a mixed approach is more cost-efficient. Note however that the mixed approach require the institution to be able to observe the population composition (i.e. the number of cooperators in the population, $j$).
\end{proposition}
The proof of this Proposition will be given in Section \ref{sec: proof of 1st result} and see Figure \ref{fig:comparison} for an illustration.

The following number is central to the analysis of this paper 
\begin{equation}
\label{eq: HNab}
H_{N,a,b}=\sum\limits_{j=1}^{N-1}\frac{1}{j(N-j)}\min(\frac{j}{a},\frac{N-j}{b}).    
\end{equation}

It plays a similar role as the harmonic number $H_N$ in \cite{duong2021cost}, where a similar cost optimisation problem but for a separate reward or punishment incentive is studied. However, unlike the harmonic function, which has a growth of $\ln N+\gamma$ (where $\gamma$ is the Euler–Mascheroni constant) as $N\rightarrow+\infty$, we will show that $H_{N,a,b}$ is always bounded and its asymptotic behaviour is given by (see more details in Proposition \ref{prop: limit HNab})
$$
H_{a,b}:=\lim\limits_{N\rightarrow +\infty} H_{N,a,b}=\frac{1}{a}\ln(\frac{a+b}{b})+\frac{1}{b}\ln(\frac{a+b}{a}).
$$

Now, our second main result below studies the asymptotic behaviour (neutral drift limit, strong selection limit, and infinite population limit) of the cost function $E_{mix}(\theta)$. 
\begin{theorem}[Asymptotic behaviour of the cost function]\
\label{thm: main thm1}
\begin{enumerate}
\item (Growth of the cost function) The cost function satisfies the following lower and upper bound estimates 
$$
\frac{N^2\theta}{2}\Big(H_{N,a,b}+\frac{1}{\max(a,b)(N-1)}\Big)\leq E_{mix}(\theta)\leq N(N-1)\theta \Big(H_{N,a,b}+\frac{1}{\min(a,b)\lfloor \frac{(N-1)}{2} \rfloor}\Big).
$$
In particular, since $H_{N,a,b}$ is uniformly bounded (with respect to $N$), it follows that the cost function grows quadratically with respect to $N$.
\item Neutral drift limit:
    $$
\lim\limits_{\beta\rightarrow 0}E_{mix}(\theta)=\theta N^2 H_{N,a,b}.
$$
\item Strong selection limit:
$$
\lim\limits_{\beta\rightarrow +\infty}E_{mix}(\theta)=\begin{cases}
\frac{N^2\theta}{2}\Big(H_{N,a,b} + \frac{1}{a(N-1)}\Big),  \quad\text{for}\quad \theta<-\delta,\\
\frac{N A}{2}\Big[2N H_{N,a,b}+\frac{1}{a(N-1)}+\frac{1}{b(N-1)}-\frac{\min(2/a, (N-2)/b)}{2(N-2)}-\frac{\min((N-1)/a,1/b)}{N-1}\Big], \quad\text{for}\quad \theta=-\delta,\\
\frac{N^2\theta}{2}\bigg[H_{N,a,b}+\frac{1}{b(N-1)}\bigg]  \quad\text{for}\quad \theta>-\delta.
\end{cases}
$$
\item Infinite population limit: 
\begin{equation}
\lim_{N\rightarrow +\infty}\frac{E_{mix}(\theta)}{\frac{N^2\theta}{2}H_{a,b}}=\begin{cases}
1+e^{-\beta|\theta-c|} \quad \text{for DG},\\
1+ e^{-\beta |\theta-c(1-\frac{r}{n})|}\quad\text{for PGG}.
\end{cases}
\end{equation}
\end{enumerate}
\end{theorem}
It is worth noticing that the neutral drift and strong selection limits of the cost function do not depend on the underlying games, but the infinite population limit does.

The proof of this Theorem will be given in Section \ref{sec: asymptotic behaviour}. Figures \ref{fig:weakstrong} and \ref{fig:infinitepopulation} provide  numerical simulations of the neutral drift and strong selection limits and the infinite population one.

The following result provides a detailed analysis for the minimisation problem \eqref{eq: min prob}  for $a=b$. Note that since $N\geq 2$, this case belongs to the interesting regime where $\frac{1}{N-1}\leq \frac{b}{a}\leq N-1$, see Proposition \ref{prop: 1st result} above. Mathematically, this case is distinguishable since it gives rise to many useful and beautiful symmetric properties and cancellations, see Section \ref{sec: reward and punishment}. We also numerically investigate the case $a\neq b$ and conjecture that the result also holds true.

\begin{theorem}[Optimisation problems and phase transition phenomenon]
\label{thm: main thm2}
\begin{enumerate}\
\item(Phase transition phenomena and behaviour under the threshold) For $a=b$, there exists a threshold value $\beta^*$ given by 
$$
\beta^*=-\frac{F^*}{\delta}>0,
$$
with 
$$
F^*=\min\{F(u):\quad u>1\},
$$
where $F(u):=\frac{Q(u)}{uP(u)}-\log (u)$ for 
\begin{align}
\label{eq: def of P}
P(u)&:=(1+u)\Bigg[\Big(\sum_{j=0}^{N-2}\Big(H_{N,a,b} + \frac{\min(\frac{j+1}{a},\frac{N-j-1}{b})}{(j+1)(N-j-1)}\Big) u^j\Big)\Big(\sum_{j=1}^{N-1} j u^{j-1}\Big)\notag
\\&\qquad\qquad-\Big(\sum_{j=1}^{N-2}\Big(H_{N,a,b} + \frac{\min(\frac{j+1}{a},\frac{N-j-1}{b})}{(j+1)(N-j-1)}\Big) j u^{j-1}\Big)\Big(\sum_{j=0}^{N-1}u^j\Big)\Bigg]\notag
\\&\qquad\qquad-\Big(\sum_{j=0}^{N-2}\Big(H_{N,a,b} + \frac{\min(\frac{j+1}{a},\frac{N-j-1}{b})}{(j+1)(N-j-1)}\Big) u^j\Big)\Big(\sum_{j=0}^{N-1}u^j\Big)
\end{align}
and
\begin{align}
\label{eq: def of Q}
    Q(u)&:=(1+u)\Big(\sum_{j=0}^{N-2}\Big(H_{N,a,b} + \frac{\min(\frac{j+1}{a},\frac{N-j-1}{b})}{(j+1)(N-j-1)}\Big) u^j\Big)\Big(\sum_{j=0}^{N-1} u^{j}\Big),
\end{align}
with $u:=e^x$, such that $\theta\mapsto E_{mix}(\theta)$ is non-decreasing for all $\beta\leq \beta^*$ and it is non-monotonic when $\beta>\beta^*$. As a consequence, for $\beta\leq \beta^*$ 
\begin{equation}
\label{eq: smallbeta}
\min\limits_{\theta\geq\theta_0}E_{mix}(\theta)=E_{mix}(\theta_0).
\end{equation}  
\item (Behaviour above the threshold value) For $\beta>\beta^*$, the number of changes of sign of $E_{mix}'(\theta)$ is at least two for all $N$ and there exists an $N_0$ such that the number of changes is exactly 2 for $N\leq N_0$. As a consequence,  
for $N\leq N_0$, there exist $\theta_1<\theta_2$ such that for $\beta>\beta^*$, $E_{mix}(\theta)$ is increasing when $\theta< \theta_1$, decreasing when $\theta_1<\theta<\theta_2$, and increasing when $\theta>\theta_2$. Thus, for $N\leq N_0$, 
$$
\min\limits_{\theta\geq\theta_0}E_{mix}(\theta)=\min\{E_{mix}(\theta_0),E_{mix}(\theta_2)\}.
$$
\end{enumerate}
\end{theorem}
The proof of this Theorem is detailed in Section \ref{sec: reward and punishment}.
\subsection{Algorithms}
Based on the results above, we describe below an algorithm for the computation of the critical threshold $\beta^*$ of the strength selection.
\begin{algorithm}[\textbf{Finding optimal cost of incentive $\theta^\star$}]\ \\ 
\label{algo1}
\textbf{Inputs}: i) $N\leq N_0$: population size, ii) $\beta$: intensity of selection, iii) game and parameters: DG ($c$ and $B$) or PGG ($c$, $r$ and $n$), iv) $\omega$: minimum desired  cooperation level, v) $a$ and $b$: reward and punishment efficiency ratios.
\begin{enumerate}[(1)]
\item Compute $\delta$ \Big\{in DG: $\delta = - (c + \frac{B}{N-1})$; in PGG: $\delta = -c \left(1 - \frac{r(N-n)}{n(N-1)} \right)$\Big\}.
\item Compute $\theta_0 = \frac{1}{(N-1)\beta} \log\left(\frac{\omega}{1-\omega}\right) - \delta$; 
\item Compute 
$$
F^*=\min\{F(u): u>1\},
$$
where $F(u)$ is defined in \eqref{eq: F}.
\item Compute $\beta^*=-\frac{F^*}{\delta}$.
\item If $\beta\leq \beta^*$:
$$
\theta^*=\theta_0,\quad \min E_{mix}(\theta)=E_{mix}(\theta_0).
$$ 
\item Otherwise (i.e. if $\beta>\beta^*$)
\begin{enumerate}
\item Compute $u_2$ that is the largest root of the equation $F(u)+\beta \delta=0$.
\item Compute $\theta_2=\frac{\log u_2}{\beta}-\delta$.
	\begin{itemize} 
		\item If $\theta_2 \leq \theta_0$: $\theta^*=\theta_0,\quad \min E_{mix}(\theta)=E_{mix}(\theta_0)$;
		\item Otherwise (if  $\theta_2 > \theta_0$): 
			\begin{itemize}  
				\item If $E_{mix}(\theta_0)\leq E_{mix}(\theta_2)$: $\theta^*=\theta_0,\quad \min E_{mix}(\theta)=E_{mix}(\theta_0)$;
				\item if $E(\theta_2)< E(\theta_0)$: $\theta^*=\theta_2,\quad \min E_{mix}(\theta)=E_{mix}(\theta_2)$.
			\end{itemize} 
	\end{itemize} 
\end{enumerate}
\end{enumerate}
\noindent \textbf{Output}: $\theta^*$ and $E_{mix}(\theta^*)$.
\end{algorithm} 

\begin{figure}[H]
\centering
\includegraphics[width=1\textwidth]{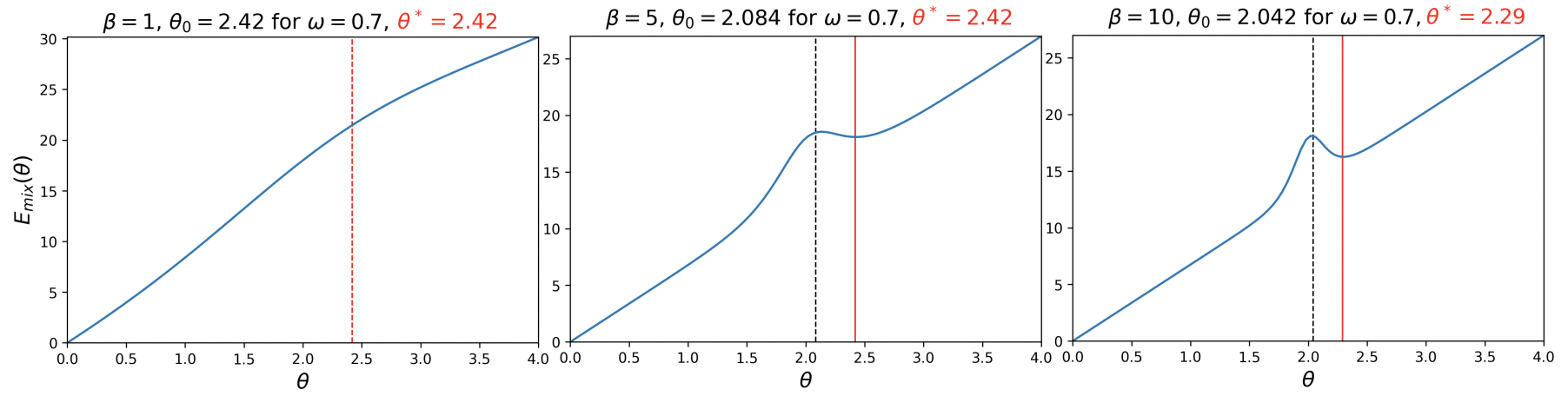}
\caption{We use Algorithm 1 to find the optimal cost per capita $\theta$, denoted by $\theta^*$ (represented by the red line in the figures), that minimises $E_{mix}(\theta)$ while ensuring a minimum level of cooperation $\omega$, where $N=3$ for DG with $B = 2, c = 1$. The left image illustrates the behaviour of the cost function for $\beta = 1$, the middle one for $\beta=5$, while the right one for $\beta = 10$. The critical threshold value for the strength of selection is $\beta^* = 3.67$ and the desired level of cooperation is $\omega=0.7$. The numerical results  obtained are in accordance with Theorem \ref{thm: main thm2}: for $\beta\leq\beta^*$, the cost function increases, while for $\beta>\beta^*$ it is not monotonic.}
\label{fig:optimalDG}
\end{figure}


\begin{figure}[H]
    \centering 
\begin{subfigure}{1\textwidth}
  \includegraphics[width=\linewidth]{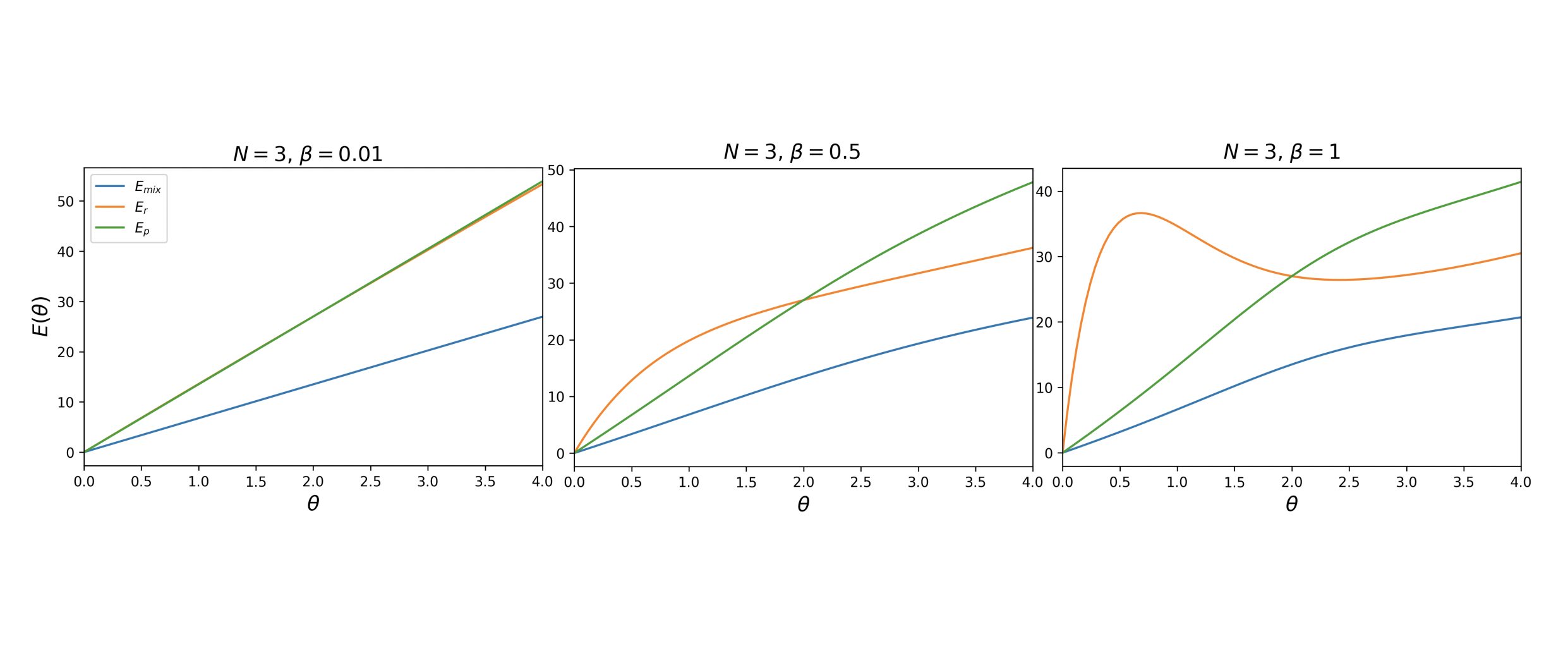}
\end{subfigure}\hfil 
\smallskip
\begin{subfigure}{1\textwidth}
  \includegraphics[width=\linewidth]{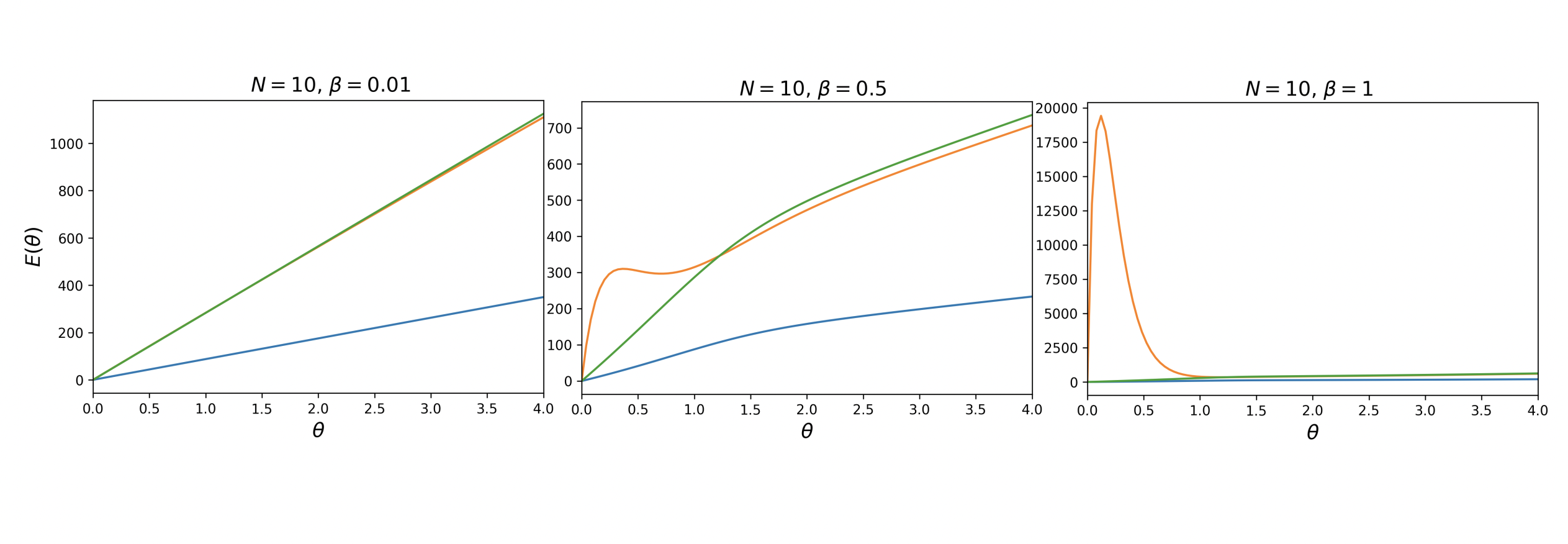}
\end{subfigure}\hfil 
\caption{Comparison of the total costs  for mixed incentives (blue), only reward (orange), and only punishment (green) as a function of the cost per capita $\theta$, for different values of $N$ and $\beta$ for DG with $B = 2, c = 1$. The first row illustrates the behaviour of the different cost functions when $N=3$ with $\beta= 0.01, 0.5, 1$ respectively, while the second row presents the comparison between the cost functions for $N=10$ with $\beta= 0.01, 0.5, 1$, respectively. As proven in Proposition \ref{prop: 1st result}, mixed incentives are less costly than either reward or punishment.}
\label{fig:comparison}
\end{figure}

\section{The expected cost function, phase transition, and the minimisation problem}
\label{sec: reward and punishment}
In this section, we study in detail the cost function, establishing the phase transition phenomena and solving the minimisation problem \ref{eq: min prob} of finding the optimal incentive, thus proving Theorem \ref{thm: main thm2}.

We consider the case 
$$
\frac{1}{N-1}<\frac{b}{a}< N-1,
$$ and focus on the most important case when $a=b$. This is due to Proposition \ref{prop: 1st result}, when $\frac{b}{a}$ is not in this interval, the mixed incentive problem reduces to either the reward or the punishment problem that has been studied in \cite{duong2021cost}.

\subsection{The cost function and its derivative}
\label{sec: cost function}
In this section, we provide the explicit computation for the cost function and its derivative. The class of cooperation dilemmas, namely DG and PGG, introduced in Section \ref{sec: models} are of crucial importance in the analysis of this paper. This is because, as already mentioned, the difference in payoffs between a cooperator and a defector, $\delta=\Pi_C(i)-\Pi_D(i)$, does not depend on the state $i$, which gives rise to explicit and analytically tractable formulas for the entries of the fundamental matrix $\mathcal{N}$ of the absorbing Markov chain describing the population dynamics. These entries are given in the following lemma, whose detailed proof can be found in \cite{duong2021cost}.
\begin{lemma}
\label{lem: N entries}
The entries $n_{1,i}$ and $n_{N-1,i}$, $i=1,\ldots, N-1$, of the fundamental matrix $\mathcal{N}$ (see its definition underneath \eqref{eq: transition probabilities}) are given by
\begin{equation*}
n_{1,i}=\frac{N^2}{(N-i)i}\frac{(e^{(N-1)x}-e^{(i-1)x})(1+e^x)}{e^{Nx}-1},\quad n_{N-1,i}=\frac{N^2}{(N-i)i}\frac{(e^{ix}-1)(1+e^x)}{e^{Nx}-1},
\end{equation*}
where $x=x(\theta):=\beta(\theta+\delta)$.
\end{lemma}
Using Lemma \ref{lem: N entries}, we obtain a more explicit formula for the cost function as follows:
\begin{align*}
E_{mix}(\theta) &= \frac{\theta}{2}\sum_{j=1}^{N-1}(n_{1j}+n_{N-1,j}) \min\Big(\frac{j}{a}, \frac{N-j}{b}\Big)\\
&=\frac{\theta}{2} \sum_{j=1}^{N-1} \frac{N^2}{(N-j)j} \Big((W^{-1})_{1j} +(W^{-1})_{N-1,j}\Big) \min\Big(\frac{j}{a}, \frac{N-j}{b}\Big) \notag\\
&= \frac{N^2 \theta(1+e^x)}{2(e^{Nx}-1)} \sum_{j=1}^{N-1} \frac{e^{(N-1)x} - e^{(j-1)x} + e^{jx} - 1 }{j(N-j)}\min\Big(\frac{j}{a}, \frac{N-j}{b}\Big)\notag
\\&=\frac{N^2 \theta(1+e^x)}{2(e^{Nx}-1)}\Big[\big(e^{(N-1)x}-1\big)\sum_{j=1}^{N-1}\frac{1}{j(N-j)}\min\Big(\frac{j}{a}, \frac{N-j}{b}\Big)+
\\&\hspace{4cm}+ (e^x-1)\sum_{j=1}^{N-1}\frac{e^{(j-1) x}}{j(N-j)}\min\Big(\frac{j}{a}, \frac{N-j}{b}\Big)\Big]\notag
\\&=\frac{N^2 \theta(1+e^x)}{2(1+e^x+\ldots e^{(N-1)x})}\Big[\big(1+e^x+\ldots +e^{(N-2)x}\big)\sum_{j=1}^{N-1}\frac{1}{j(N-j)}\min\Big(\frac{j}{a}, \frac{N-j}{b}\Big)+
\\&\hspace{8cm}+\sum_{j=1}^{N-1}\frac{e^{(j-1) x}}{j(N-j)}\min\Big(\frac{j}{a}, \frac{N-j}{b}\Big)\Big].
\end{align*}

To study the monotonicity of the cost function, in the following lemma, we compute the derivative of $E_{mix}(\theta)$ with respect to $\theta$.
\begin{lemma}
\label{lem: Emix and derivative}
The total cost of interference for the mixed institutional incentive is given by
\begin{align*}
E_{mix}(\theta)&=\frac{N^2 \theta(1+e^x)}{2(1+e^x+\ldots e^{(N-1)x})}\Big[\big(1+e^x+\ldots +e^{(N-2)x}\big)\sum_{j=1}^{N-1}\frac{1}{j(N-j)}\min\Big(\frac{j}{a}, \frac{N-j}{b}\Big)
\\&\qquad+\sum_{j=1}^{N-1}\frac{e^{(j-1) x}}{j(N-j)}\min\Big(\frac{j}{a}, \frac{N-j}{b}\Big)\Big].
\end{align*}
Its derivative is given by
\begin{align*}
E_{mix}'(\theta)=\frac{N^2}{2g(x)^2}\left[Q(u)-(x-\beta \delta)uP(u)\right],
\end{align*}
where $P(u)=f(x)g'(x)-f'(x)g(x)$ and $Q(u)=f(x)g(x)$ for $f(x)=(1+u)\sum_{j=0}^{N-2} u^j\Big(H_{N,a,b} + \frac{\min(\frac{j+1}{a},\frac{N-j-1}{b})}{(j+1)(N-j-1)}\Big)$ and $g(x)=\sum_{j=0}^{N-1} u^j$ with $u=e^x$.
\end{lemma}

See Section \ref{sec: prooflem2} for a proof of this lemma.

\subsection{The polynomial $P$}
This section contains details about 
\begin{align*}
P(u)&:=(1+u)\Bigg[\Big(\sum_{j=0}^{N-2}\Big(H_{N,a,b} + \frac{\min(\frac{j+1}{a},\frac{N-j-1}{b})}{(j+1)(N-j-1)}\Big) u^j\Big)\Big(\sum_{j=1}^{N-1} j u^{j-1}\Big)\notag
\\&\qquad\qquad-\Big(\sum_{j=1}^{N-2}\Big(H_{N,a,b} + \frac{\min(\frac{j+1}{a},\frac{N-j-1}{b})}{(j+1)(N-j-1)}\Big) j u^{j-1}\Big)\Big(\sum_{j=0}^{N-1}u^j\Big)\Bigg]\notag
\\&\qquad\qquad-\Big(\sum_{j=0}^{N-2}\Big(H_{N,a,b} + \frac{\min(\frac{j+1}{a},\frac{N-j-1}{b})}{(j+1)(N-j-1)}\Big) u^j\Big)\Big(\sum_{j=0}^{N-1}u^j\Big).
\end{align*}

The following proposition studies the properties of $P$. 
\begin{proposition}
\label{prop: coeffs}
Let $P(u)$ be the polynomial defined in \eqref{eq: def of P}.  Then it is a polynomial of degree $2N-4$,
$$
P(u)= \sum_{k=0}^{2N-4} p_k u^k,
$$
where the leading coefficient $p_{2N-4}$ is positive. When $a=b$, the coefficients of $P$ are anti-symmetric, that is we have
$$
p_k=-p_{2N-4-k}<0 \quad\text{for}\quad k=0,\ldots, N-3, \quad \text{and}\quad p_{N-2}=0.
$$
As a consequence, when $a=b$, $P$ has exactly one positive root,  which is equal to $1$.  
\end{proposition}
The proof of this proposition is lengthy and delicate. The case $a=b$ is special since it gives rise to many useful symmetric properties and nice cancellations. To focus on the main points here, we postpone the proof to the Appendix, see Section \ref{sec:proof1}.

\subsection{The derivative of $F$}
In this section, we study the derivative of the function $F$ defined in \eqref{eq: F}. The analysis of this section will play an important role in the study of the phase transition of the cost function in the next section.

We have
$$
F'(u)=\frac{u Q'(u)P(u)-Q(u)(P(u)+uP'(u))}{u^2P(u)^2}-\frac{1}{u}=:\frac{M(u)}{u^2P(u)^2},
$$
where
\begin{equation}
\label{eq: M}
M(u):=u Q'(u)P(u)-Q(u)(P(u)+uP'(u))-uP(u)^2.
\end{equation}

The sign of $F'(u)$ (thus, the monotonicity of $F$) is the same as that of the polynomial $M$. The next proposition presents some properties of $M$.
\begin{proposition} 
\label{prop: M}
The following statements hold
\begin{enumerate}
    \item $M$ is a polynomial of degree $4N-6$,
$$
M(u)=\sum_{i=0}^{4N-6} m_i u^{i},
$$
where the leading coefficient is
$$
m_{4N-6}=a_{N-2}a_{N-3}>0.
$$
\item When $a=b$, the coefficients of $M$ are symmetric, that is for all $i=0,\ldots, 4N-6$
$$
m_{i}=m_{4N-6-i}.
$$
\item $M$ has at least two positive roots, one is less than 1 and the other is bigger than 1. For sufficiently small $N$, namely $N\leq N_0$, $M$ has exactly two positive roots, $u_1$ and $u_2$, where $u_1<1<u_2$. As a consequence, for $1<u<u_2$, $F'(u)<0$, thus $F$ is decreasing. While for $u_2<u$, $F'(u)>0$, thus $F$ is increasing.
\end{enumerate}
\end{proposition}
The proof of this proposition is presented in Section \ref{sec:proof2}.  We conjecture that the sequence of $M$ has exactly two changes of signs, and thus $M$ has exactly two positive roots. 

\subsection{The phase transition and the minimisation problem}
In this section, we study the phase transition problem, which describes the change in the behaviour of the cost function when varying the strength of selection $\beta$, and the optimal incentive problem, thus proving Theorem \ref{thm: main thm2}. We focus on the case $a=b$.
\begin{proof}[Proof of Theorem \ref{thm: main thm2}]

It follows from Proposition \ref{prop: coeffs} that for $0<u$,
$P(u)>0$ if and only if $u>1$.

Thus according to the argument at the end of Section \ref{sec: cost function}, if $u\leq 1$ then $E'_{mix}>0$ (thus $E_{mix}$ is increasing); and for $u> 1$ we have (see \eqref{eq: derivative Emix})
$$
E_{mix}'(\theta)=\frac{N^2}{2g(x)^2}(uP(u))\Big(\frac{Q(u)}{uP(u)}-\log(u) + \beta \delta\Big)=\frac{N^2}{2g(x)^2}(uP(u))(F(u) + \beta \delta),   
$$
where the function $F$ is (see \eqref{eq: F})
$$
F(u)= \frac{Q(u)}{uP(u)}-\log(u).
$$
Since $Q(u)$ is a polynomial of degree $2N-2$ and $uP(u)$ is a polynomial of degree $2N-3$ and their leading coefficients are both positive,
$$
\lim_{u\rightarrow +\infty} F(u)=+\infty=\lim_{u\rightarrow 1^+} F(u).
$$
This, together with the fact that $F$ is smooth on $(1,+\infty)$, we deduce that there exists a global minimum of $F$ in the interval $(1,+\infty)$ 
$$
F^*:=\min\{F(u),~u>1\}.
$$
Let
$$
\beta^*:=-\frac{F^*}{\beta \delta}=-\frac{F^*}{\delta}.
$$
Then it follows from  the above formula of $E'_{mix}(\theta)$ that for $\beta\leq \beta^*$, $E_{mix}'(\theta)\geq 0$. Thus for $\beta\leq \beta^*$, $E_{mix}(\theta)$ is always increasing. For $\beta>\beta^*$, the sign of $E_{mix}'(\theta)$ depends on the sign of the term $F(u)+\beta \delta$. For arbitrary $N$, the equation $F(u)=-\beta \delta$ has at least two roots, thus the sign of the term $F(u)+\beta \delta$ changes at least twice, therefore $E$ is not monotonic. In particular, for $N\leq N_0$, since $F$ is decreasing in $(1,u_2)$ and increasing in $(u_2,+\infty)$, where $F^*=F(u_2)<-\beta \delta$, the equation $F(u)=-\beta \delta$ has two roots $1<\bar{u_1}<u_2<\bar{u}_2$, and $F(u)+\beta \delta<0$ when $\bar{u}_1< u< \bar{u}_2$ while $F(u)+\beta \delta\geq 0$ when $u\in (1,\bar{u}_1]\cup [\bar{u}_2,+\infty)$, see Figure \ref{fig:F} for an illustration.

\begin{figure}[H]
    \centering
    \includegraphics[scale=0.5]{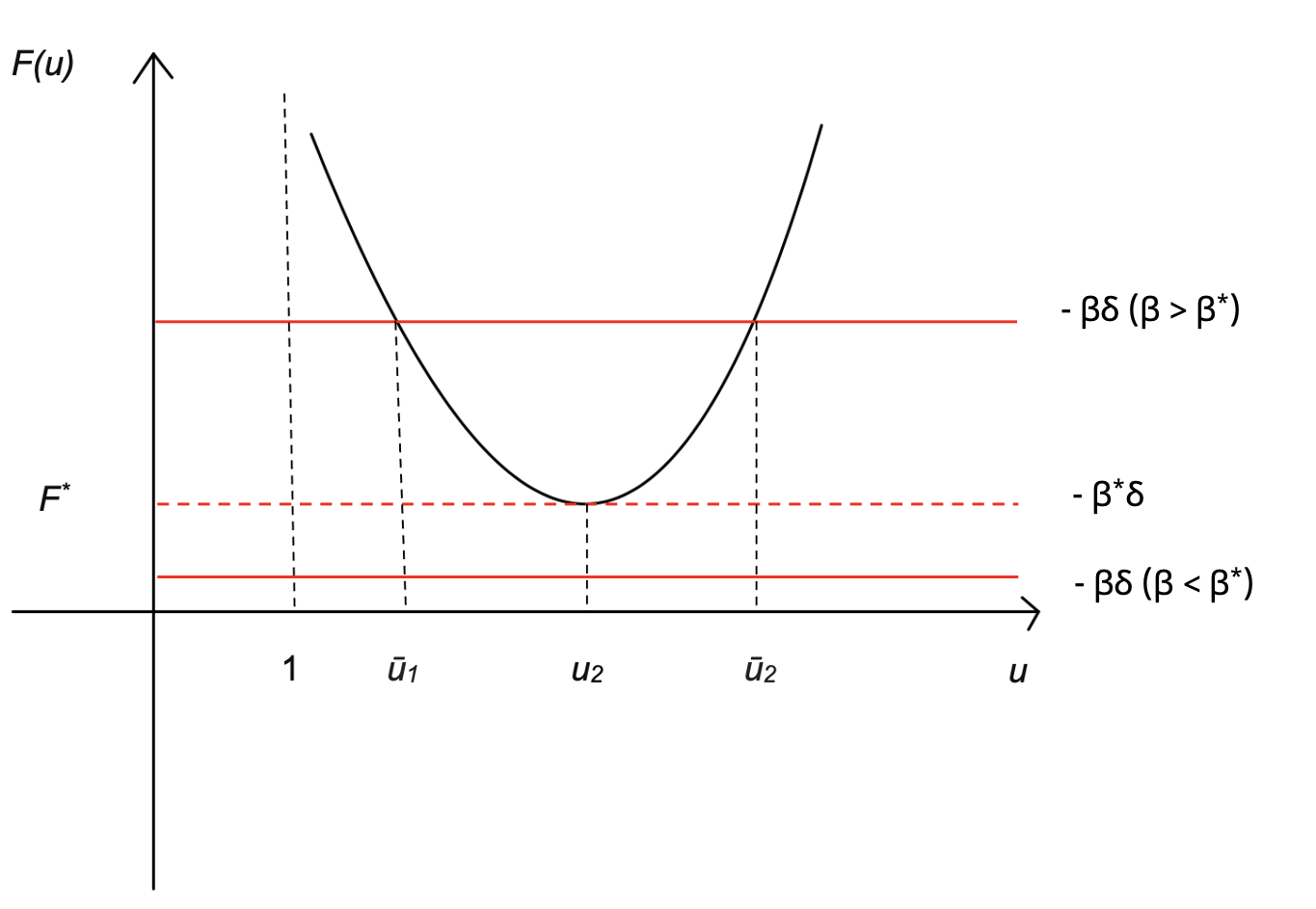}
    \caption{Behaviour of $F$ and determination of the critical threshold $\beta^*$. For sufficiently small N, i.e. $N\leq N_0$, since $F$ is decreasing in $(1,u_2)$ and increasing in $(u_2,+\infty)$, where $F^*=F(u_2)<-\beta \delta$, the equation $F(u)=-\beta \delta$ has two roots $1<\bar{u_1}<u_2<\bar{u}_2$, and $F(u)+\beta \delta<0$ when $\bar{u}_1< u< \bar{u}_2$ while $F(u)+\beta \delta\geq 0$ when $u\in (1,\bar{u}_1]\cup [\bar{u}_2,+\infty)$.}
    \label{fig:F}
\end{figure}

Hence $E_{mix}'(\theta)<0$ when $\bar{u}_1< u< \bar{u}_2$ while $E'_{mix}(\theta)\geq 0$ when $u\in (1,\bar{u}_1]\cup [\bar{u}_2,+\infty)$. Thus $E_{mix}(\theta)$ is increasing in $(1,\bar{u}_1)$, decreasing in $(\bar{u}_1,\bar{u}_2)$, and increasing again in $(\bar{u}_2,+\infty)$. In term of the variable $\theta$, $E_{mix}(\theta)$ is increasing in $(-\delta,\theta_1)$, decreasing in $(\theta_1,\theta_2)$, and increasing again in $(\theta_2,+\infty)$, where
$$
\theta_1:=\frac{\log(u_1)}{\beta
} - \delta,\quad \theta_2:=\frac{\log(u_2)}{\beta}- \delta.
$$
As a consequence, for $N\leq N_0$,
$$
\min_{\theta\geq \theta_0}
E_{mix}(\theta)=\min\{E_{mix}(\theta_0), E_{mix}(\theta_2)\}.
$$

This completes the proof of this theorem.
\end{proof}

\section{Asymptotic behaviour of the expected cost function}
\label{sec: asymptotic behaviour}
In this section, we study the asymptotic behaviour (neutral drift, strong selection, and infinite population limits) of the cost function, proving the main results in Theorem \ref{thm: main thm1}. To this end, we will first need some auxiliary technical lemmas. 
\subsection{Some auxiliary lemmas}
The first auxiliary lemma is an elementary inequality which will be used to estimate the cost function.  Its proof can be found in \cite{duong2021cost}.
\begin{lemma} For all $x\in\mathbb{R}$, we have
\label{lem: max}
\begin{equation}
0\leq \frac{e^x+\ldots+e^{(N-2)x}}{1+e^x+\ldots+e^{(N-1)x}}\leq \frac{N-2}{N}.
\end{equation}
\end{lemma}
In the next lemma, we provide lower and upper bounds for the number $H_{N,a,b}$ defined in \eqref{eq: HNab}. As mentioned in the introduction, this number plays a similar role as the harmonic number $H_n$ in \cite{duong2021cost}. However, unlike the harmonic number, we will show that $H_{N,a,b}$ is always bounded and has a finite limit as $N\rightarrow+\infty$.  \begin{lemma}
\label{lem: bound of HNab}
It holds that
$$
\frac{2(\ln{2}+\frac{1}{2N-1}-\frac{1}{N+1})}{\max(a,b)}\leq H_{N,a,b}\leq \frac{2(\ln{2}+\frac{1}{2N+1}-\frac{1}{N-1})}{\min(a,b)}.
$$ 
\end{lemma}

The proof of this lemma is given in Section \ref{sec: proof3}.\\

The following proposition characterises the asymptotic limit of $H_{N,a,b}$ as  $N\rightarrow+\infty$, which will be used later to obtain asymptotic limits of the cost function.
\begin{proposition}
\label{prop: limit HNab}
It holds that
$$
\lim\limits_{N\rightarrow +\infty} H_{N,a,b}=H_{a,b}\quad\text{,where}\quad H_{a,b}=\frac{1}{a}\ln(\frac{a+b}{b})+\frac{1}{b}\ln(\frac{a+b}{a}).
$$
\end{proposition}

\begin{proof}
We recall that $H_{N,a,b}=\sum\limits_{j=1}^{N-1}\frac{1}{j(N-j)}\min(\frac{j}{a},\frac{N-j}{b})$. As in the proof of the above lemma, we will link $H_{N,a,b}$ to the harmonic number $H_N$ by splitting the sum at an appropriate point, which allows us to determine the minimum between $j/a$ and $(N-j)/a$ explicitly, given by 
$$
N_{a,b}=\left\lfloor \frac{N}{\frac{b}{a}+1} \right\rfloor.
$$
We have 
\begin{align*}
    H_{N,a,b}= \sum\limits_{j=1}^{N_{a,b}}\frac{\frac{j}{a}}{j(N-j)} + \sum\limits_{j=N_{a,b}+1}^{N-1}\frac{\frac{N-j}{b}}{j(N-j)} = \frac{1}{a}\sum\limits_{j=1}^{N_{a,b}}\frac{1}{N-j} + \frac{1}{b}\sum\limits_{j=N_{a,b}}^{N-1}\frac{1}{j}.
\end{align*}
\noindent Let $\hat{j}=N-j$. Thus:
\begin{align*}
    H_{N,a,b}= \frac{1}{a}\sum\limits_{\hat{j}=N-N_{a,b}}^{N-1}\frac{1}{\hat{j}} + \frac{1}{b}\sum\limits_{j=N_{a,b}+1}^{N-1}\frac{1}{j}= \frac{1}{a}\Big(H_{N}-H_{N-N_{a,b}}\Big) + \frac{1}{b}\Big(H_{N}-H_{N_{a,b}}\Big).
\end{align*}

Note that $N-N_{a,b}=N(1-\frac{a}{a+b})=N\frac{b}{a+b}$ and $N_{a,b}=N\frac{a}{a+b}.$
Thus both $N_{a,b}$ and $N-N_{a,b}$ go to $+\infty$ as $N\rightarrow +\infty$. Then it follows from the following well-known asymptotic behaviour of the harmonic number (see \eqref{eq: HN estimates})
$$
\lim\limits_{N\rightarrow+\infty}H_N=\ln N+\gamma,
$$
we obtain the following asymptotic behaviour of $H_{N,a,b}$
$$
\lim\limits_{N\rightarrow +\infty} H_{N,a,b}=\frac{1}{a}\ln(\frac{a+b}{b})+\frac{1}{b}\ln(\frac{a+b}{a}).
$$ 
This completes the proof of the proposition.
\end{proof}
The following lemma provides lower and upper bounds for the cost function, which show that it grows quadratically with respect to the population size.
\begin{lemma}
It holds that
\begin{equation}
\label{eq: estimates1}
\frac{N^2\theta}{2}\Big(H_{N,a,b}+m\Big)\leq E_{mix}(\theta)\leq N(N-1)\theta \Big(H_{N,a,b}+M\Big),
\end{equation}
where
\begin{equation}
\label{eq: m and M}
m=\min_{i}\frac{\min(\frac{i+1}{a},\frac{N-i-1}{b})}{(i+1)(N-i-1)},\quad  M=\max_{i}\frac{\min(\frac{i+1}{a},\frac{N-i-1}{b})}{(i+1)(N-i-1)}.
\end{equation}
\end{lemma}
\begin{proof}
Let $H_{N,a,b}=\sum\limits_{j=1}^{N-1}\frac{1}{j(N-j)}\min(\frac{j}{a},\frac{N-j}{b})$.
We have
$$
\frac{(1+e^x)(1+e^x+\ldots+e^{(N-2)x})}{1+e^x+\ldots+e^{(N-1)x}}=1+\frac{e^x+\ldots+e^{(N-2)x}}{1+e^x+\ldots+e^{(N-1)x}}.
$$

Using Lemma \ref{lem: max}, we get
$$
1\leq \frac{(1+e^x)(1+e^x+\ldots+e^{(N-2)x})}{1+e^x+\ldots+e^{(N-1)x}}\leq \frac{2(N-1)}{N}.
$$

Let $m=\min\limits_{i}\frac{\min(\frac{i+1}{a},\frac{N-i-1}{b})}{(i+1)(N-i-1)}$ and $M=\max\limits_{i}\frac{\min(\frac{i+1}{a},\frac{N-i-1}{b})}{(i+1)(N-i-1)}$.
Since
\begin{equation*}
m\sum_{j=0}^{N-2}e^{jx}\leq \sum_{j=1}^{N-1}\frac{e^{(j-1) x}}{j(N-j)}\min\Big(\frac{j}{a},\frac{N-j}{b}\Big)\leq M\sum_{j=0}^{N-2}e^{jx}.
\end{equation*} 
we obtain the following estimates
\begin{align*}
m\leq m\frac{(1+e^x)(1+e^x+\ldots+e^{(N-2)x})}{1+e^x+\ldots+e^{(N-1)x}}&\leq \frac{(1+e^x)}{1+e^x+\ldots+e^{(N-1)x}}\sum_{j=1}^{N-1}\frac{e^{(j-1)x}}{j(N-j)}\min\Big(\frac{j}{a},\frac{N-j}{b}\Big)
\\&\leq (H_{N,a,b}+M)\frac{(1+e^x)(1+e^x+\ldots+e^{(N-2)x})}{1+e^x+\ldots+e^{(N-1)x}}\leq \frac{2(N-1)(H_{N,a,b}+M)}{N}.
\end{align*}

Thus for $\theta>0$ we have
$$
\frac{N^2\theta}{2}\Big(H_{N,a,b}+m\Big)\leq E_{mix}(\theta)\leq N(N-1)\theta \Big(H_{N,a,b}+M\Big),
$$
which completes the proof of the lemma.
\end{proof}
We can further estimate $m$ and $M$ in \eqref{eq: m and M} from below and above respectively in terms of the population size $N$ as follows.
\begin{lemma}
\label{lem: m and M}
For $m$ and $M$ defined in \eqref{eq: m and M}, it holds that 
$$
m\geq \frac{1}{\max(a,b)(N-1)},\quad M\leq \frac{1}{\min(a,b)\lfloor \frac{(N-1)}{2} \rfloor}.
$$
As a consequence, we have
$$
\frac{N^2\theta}{2}\Big(H_{N,a,b}+\frac{1}{\max(a,b)(N-1)}\Big)\leq E_{mix}(\theta)\leq N(N-1)\theta \Big(H_{N,a,b}+\frac{1}{\min(a,b)\lfloor \frac{(N-1)}{2} \rfloor}\Big).
$$
\end{lemma}
The proof of this lemma is given in Section \ref{sec: proof4}.
\subsection{Neutral drift limit}
In this section, we study the neutral drift limit of the cost function, proving the second part of Theorem \ref{thm: main thm1}.
\begin{proposition}[neutral drift limit] It holds that
\label{prop: weak selection}
$$
\lim\limits_{\beta\rightarrow 0}E_{mix}(\theta)=\theta N^2 H_{N,a,b}.
$$
\end{proposition}
It is worth noting that the neutral drift limit of the cost function depends on the population size $N$, the reward and punishment efficiency ratios $a$ and $b$ through the number $H_{N,a,b}$, but does not depend on the underlying game.
\begin{proof}
\noindent We recall that $x=\beta(\theta+\delta)$. Since $\beta\rightarrow 0$ implies $x\rightarrow 0$, it follows from the formula of $E_{mix}(\theta)$ that
\begin{align*}
    \lim_{\beta\rightarrow 0}E_{mix}(\theta)&=\lim_{\beta\rightarrow 0} \frac{N^2 \theta(1+e^x)}{2(1+e^x+\ldots e^{(N-1)x})}\bigg[\big(1+e^x+\ldots +e^{(N-2)x}\big)\sum_{j=1}^{N-1}\frac{1}{j(N-j)}\min\Big(\frac{j}{a}, \frac{N-j}{b}\Big)
    \\&\hspace{6cm}+\sum_{j=1}^{N-1}\frac{e^{(j-1) x}}{j(N-j)}\min\Big(\frac{j}{a}, \frac{N-j}{b}\Big)\bigg]\notag
   \\&=\theta N[(N-1)H_{N,a,b}+H_{N,a,b}]
  \\& =\theta N^2H_{N,a,b}.
\end{align*}
\end{proof}

\subsection{Strong selection limits}
In this section, we study the strong selection limits of the cost function, proving the third statement of Theorem \ref{thm: main thm1}.
\begin{proposition} [strong selection limits]
\label{prop: strong selection}
$$
\lim\limits_{\beta\rightarrow +\infty}E_{mix}(\theta)=\begin{cases}
\frac{N^2\theta}{2}\Big(H_{N,a,b} + \frac{1}{a(N-1)}\Big),  \quad\text{for}\quad \theta<-\delta,\\
\frac{N A}{2}\Big[2N H_{N,a,b}+\frac{1}{a(N-1)}+\frac{1}{b(N-1)}-\frac{\min(2/a, (N-2)/b)}{2(N-2)}-\frac{\min((N-1)/a,1/b)}{(N-1)}\Big], \quad\text{for}\quad \theta=-\delta,\\
\frac{N^2\theta}{2}\bigg[H_{N,a,b}+\frac{1}{b(N-1)}\bigg]  \quad\text{for}\quad \theta>-\delta.
\end{cases}
$$
\end{proposition}
Similarly to the neutral drift limit, the strong selection limit of the cost function depends on the population size $N$, the reward and punishment efficiency ratios $a$ and $b$, but does not depend on the underlying game.
\begin{proof}
To establish the strong selection limit $\lim\limits_{\beta\rightarrow+\infty} E_{mix}(\theta)$, we rewrite $E_{mix}(\theta)$ in a more convenient form as follows:
\begin{align*}
E_{mix}(\theta)&=\frac{\frac{N^2\theta}{2}(1+e^x)}{\sum\limits_{j=1}^{N-1}e^{jx}}\Big[\sum\limits_{j=0}^{N-2} e^{jx}H_{N,a,b} + \sum\limits_{j=0}^{N-2} \frac{e^{\hat{j}x}}{(\hat{j}+1)(N-\hat{j}-1)\min\Big(\frac{\hat{j}+1}{a}, \frac{N-\hat{j}-1}{b}\Big)}\Big]
\\&=\frac{\frac{N^2\theta}{2}(1+e^x)}{\sum\limits_{j=0}^{N-1}e^{jx}}\Big[ \sum\limits_{j=0}^{N-2} e^{jx}(H_{N,a,b} + h_j)\Big],
\end{align*} where  $h_j=\frac{\min(\frac{j+1}{a}, \frac{N-j-1}{b})}{(j+1)(N-j-1)}$.

 Now, note that:

\begin{align*}
(1+e^x)\sum\limits_{j=0}^{N-2} e^{jx}(H_{N,a,b} + h_j)
&= \sum\limits_{j=0}^{N-2} e^{jx}(H_{N,a,b} + h_j) + \sum\limits_{j=0}^{N-2} e^{(j+1)x}(H_{N,a,b} + h_j)
\\&= \sum\limits_{j=0}^{N-1}\eta_j e^{jx},
\end{align*}
where 
\begin{align*}
&\eta_0=H_{N,a,b}+h_0=H_{N,a,b}+\frac{\min(\frac{1}{a}, \frac{N-1}{b})}{N-1}=H_{N,a,b}+\frac{1}{a(N-1)},\\
&\eta_{N-1}=H_{N,a,b}+h_{N-2}=H_{N,a,b}+\frac{1}{b(N-1)},\\
&\eta_j=H_{N,a,b}+h_j+H_{N,a,b}+h_{j-1}=2H_{N,a,b}+h_j+h_{j-1}, \quad\text{for}\quad 1\leq j\leq N-2.
\end{align*}
By putting all of the above together, we obtain that:

\begin{equation*}
E_{mix}(\theta) = \frac{N^2\theta}{2\sum\limits_{j=0}^{N-1}e^{jx}}\sum\limits_{j=0}^{N-1}\eta_j e^{jx}.
\end{equation*}
\noindent Recall that $x=\beta(\theta+\delta)$ with $\delta$ being the difference of the average payoffs between a cooperator and a defector. We study 3 cases:
\begin{enumerate}
    \item If $\theta<-\delta$, then $j(\theta+\delta)<0$, so $e^{jx}=\Big[e^{j(\theta+\delta)}\Big]^{\beta}$ for all $1\leq j \leq N-1$. Thus,
    $\lim\limits_{\beta\rightarrow+\infty} e^{jx}=\lim\limits_{\beta\rightarrow+\infty} \Big[e^{j(\theta+\delta)}\Big] = 0$. Therefore,
    \begin{align*}
        \lim\limits_{\beta\rightarrow +\infty}E_{mix}(\theta) &=\frac{N^2\theta}{2}\frac{1}{(1+0)}\eta_0
        \\&=\frac{N^2\theta}{2}\bigg(H_{N,a,b}+\frac{1}{a(N-1)}\bigg).
    \end{align*}
    
    \item If $\theta=-\delta$, then $x=0$. So $E_{mix}(\theta)=E_{mix}(-\delta)=-\frac{N^2\delta}{2N}\sum\limits_{j=0}^{N-1} \eta_j$. 
We have
\begin{align*}
 \sum\limits_{j=0}^{N-1} \eta_j &=\eta_0+ \sum_{j=1}^{N-2}\eta_j+\eta_{N-1}
 \\&=H_{N,a,b}+\frac{1}{a(N-1)}+H_{N,a,b}+\frac{1}{b(N-1)}+\sum_{j=1}^{N-2}\Big(2 H_{N,a,b}+h_j+h_{j-1}\Big)
 \\&=2(N-1) H_{N,a,b}+\frac{1}{a(N-1)}+\frac{1}{b(N-1)}+H_{N,a,b}-h_1+H_{N,a,b}-h_{N-2}
 \\&=2N H_{N,a,b}+\frac{1}{a(N-1)}+\frac{1}{b(N-1)}-\frac{\min(2/a, (N-2)/b)}{2(N-2)}-\frac{\min((N-1)/a,1/b)}{(N-1)}.
 \end{align*}
    Therefore, 
$$
E_{mix}(-\delta)=-\frac{N \delta}{2}\Big[2N H_{N,a,b}+\frac{1}{a(N-1)}+\frac{1}{b(N-1)}-\frac{\min(2/a, (N-2)/b)}{2(N-2)}-\frac{\min((N-1)/a,1/b)}{(N-1)}\Big].
$$
\item If $\theta>-\delta$, then we obtain
        \begin{align*}
         \lim\limits_{\beta\rightarrow+\infty} E_{mix}(\theta) &= \frac{N^2\theta}{2} \lim\limits_{\beta\rightarrow+\infty} \frac{ \sum\limits_{j=0}^{N-1} \eta_j e^{jx}}{\sum\limits_{j=0}^{N-1} e^{jx}}
         \\&=\frac{N^2\theta}{2}\eta_{N-1}
         \\&=\frac{N^2\theta}{2}\bigg[H_{N,a,b}+\frac{1}{b(N-1)}\bigg],
    \end{align*} since 
    \begin{align*}
        \lim\limits_{\beta\rightarrow+\infty} \frac{ \sum\limits_{j=0}^{N-1} \eta_j e^{jx}}{\sum\limits_{j=0}^{N-1} e^{jx}} &=\frac{e^{(N-1)x}\sum\limits_{j=0}^{N-1} e^{(j-(N-1))x}}{e^{(N-1)x}\sum\limits_{j=0}^{N-1} e^{(j-(N-1))x}}
        \\&=\frac{\sum\limits_{j=0}^{N-1} \eta_j e^{(j-(N-1))x}}{\sum\limits_{j=0}^{N-1} e^{(j-(N-1))x}},
    \end{align*} so $e^{(j-(N-1))x}=e^{(j-(N-1))\beta(\theta+\delta)}=\bigg[ e^{(j-(N-1))(\theta+\delta)}\bigg]^{\beta}$ for all $0\leq j \leq N-2$.\\
    Thus $\lim\limits_{\beta\rightarrow+\infty} e^{(j-(N-1))x}=0$  for all $0\leq j \leq N-2$.
    
\end{enumerate}

\end{proof}

\begin{figure}[H]
\centering
  \includegraphics[width=1.1\textwidth]{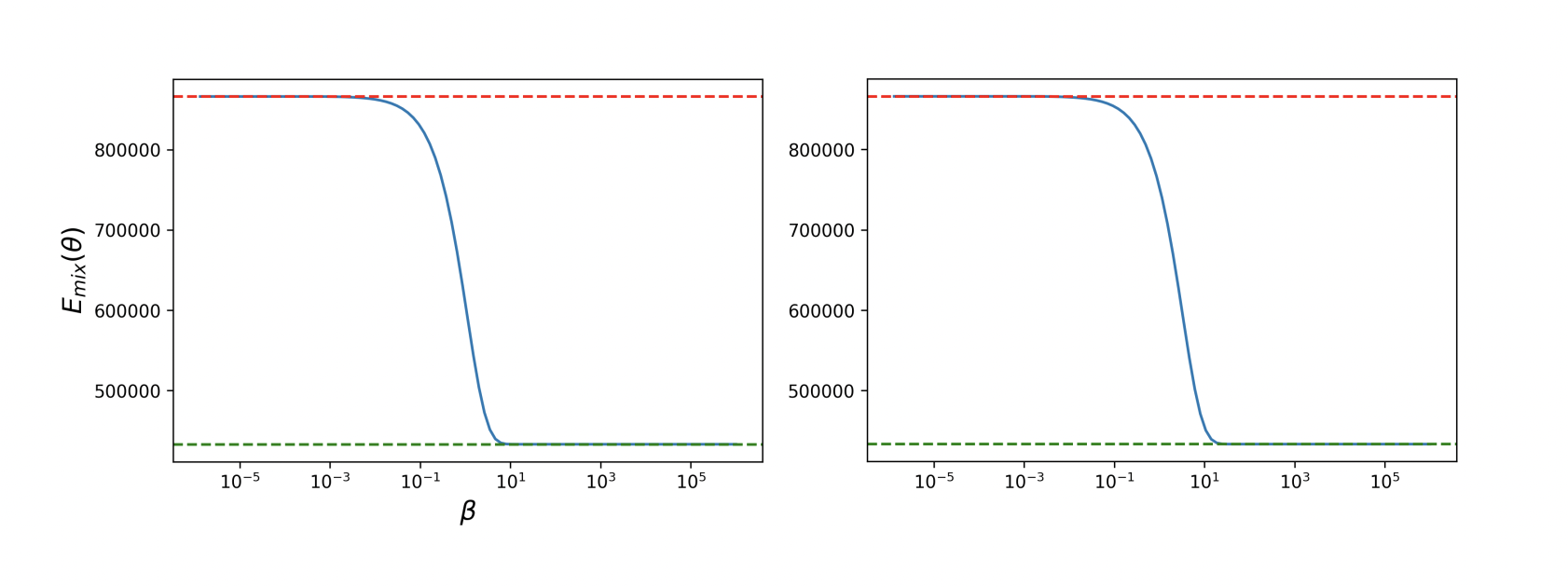}
\caption{Neutral drift and strong selection limits. We calculate numerically the expected total cost for the mixed incentive, varying the intensity of selection $\beta$. The dashed red lines represent the corresponding theoretical limiting values obtained in Proposition \ref{prop: weak selection} for the neutral drift limit, while the dashed green lines represent the corresponding theoretical limiting values obtained in Proposition \ref{prop: strong selection} for the strong selection limit. We observe that numerical results are in close accordance with those obtained theoretically. Results are obtained for DG with $N = 2500,  B = 2, c = 1$ (left) and for PGG with $N = 2500, c = 1, r = 3, n = 5$ (right).}
\label{fig:weakstrong}
\end{figure}

\subsection{Infinite population limits}
In this section, we establish the infinite population limits of the cost function, proving the fourth statement of Theorem \ref{thm: main thm1}.
\begin{proposition}[Infinite population limit] 
\label{prop: large N limit reward}
We have 
\begin{equation}
\lim_{N\rightarrow +\infty}\frac{E_{mix}(\theta)}{\frac{N^2\theta}{2}H_{a,b}}=\begin{cases}
1+e^{-\beta|\theta-c|} \quad \text{for DG},\\
1+ e^{-\beta |\theta-c(1-\frac{r}{n})|}\quad\text{for PGG}.
\end{cases}
\end{equation}
\end{proposition}
Unlike the neutral drift and strong selection limits, the infinite population limit of the cost function strongly depends on the underlying game. The strength of selection, $\beta$, is fixed. See Figure \ref{fig:infinitepopulation} below for an illustration of these limits.

\begin{proof}
\noindent Using 
$$
E_{mix}(\theta)=\frac{\frac{N^2\theta}{2}(1+e^x)}{\sum\limits_{j=0}^{N-1}e^{jx}}\Big[ \sum\limits_{j=0}^{N-2} e^{jx}(H_{N,a,b} + h_j)\Big]\quad \text{where}\quad h_j=\frac{\min(\frac{j+1}{a}, \frac{N-j-1}{b})}{(j+1)(N-j-1)},
$$

we can estimate

\begin{equation}
\label{eq: es0}
\frac{(1+e^x)(1+\ldots+e^{(N-2)x})}{1+\ldots+e^{(N-1)x}}\Big(\frac{H_{N,a,b}}{H_{a,b}}+\frac{\underline{h}}{H_{a,b}}\Big)\leq \frac{E_{mix}(\theta)}{\frac{N^2\theta}{2} H_{a,b}}\\ \leq \frac{(1+e^x)(1+\ldots+e^{(N-2)x})}{1+\ldots+e^{(N-1)x}}\Big(\frac{H_{N,a,b}}{H_{a,b}}+\frac{\overline{h}}{H_{a,b}}\Big),
\end{equation}
where $\underline{h}=\min h_j$ and $\overline{h}=\max h_j$.

We recall that $x=\beta(\theta+\delta)$, where
$$
\delta=\delta(N)=\begin{cases}
-(c+\frac{B}{N-1})\quad\text{for DG},\\
-c(1-\frac{r(N-n)}{n(N-1)}) \quad\text{for PGG}.
\end{cases}
$$

Let
$$
w(\theta):=\frac{(1+e^x)(1+\ldots+e^{(N-2)x})}{1+\ldots+e^{(N-1)x}}.
$$

Then we have
\begin{equation}
\label{eqn: omega}
w(\theta)=1+\frac{e^x+\ldots e^{(N-2)x}}{1+\ldots+e^{(N-1)x}}=\begin{cases}
\frac{2(N-1)}{N}\quad\text{if}\quad x=0,\\
1+e^x\frac{1-e^{(N-2)x}}{1-e^{N x}}\quad \text{if}\quad x\neq 0.
\end{cases}
\end{equation}

Since for fixed $\theta$, $x$ equals to $0$ for only one value of $N$, we can consider $x\neq 0$ when we study the limit $N\rightarrow +\infty$. We have 
\begin{equation}
\label{eqn: exponential}
\lim_{N\rightarrow +\infty} e^x=\lim_{N\rightarrow +\infty}e^{\beta(\theta+\delta(N))}=\begin{cases}
e^{\beta(\theta-c)}\quad\text{for DG},\\
e^{\beta(\theta-c(1-\frac{r}{n}))}\quad \text{for PGG}.
\end{cases}
\end{equation}
In addition, for DG, we have
\begin{equation}
\label{eq: DG1}
\lim_{N\rightarrow +\infty}\frac{1-e^{(N-2)x}}{1-e^{N x}}=\lim_{N\rightarrow +\infty}\frac{1-e^{-\beta b(N-3)/(N-1)}e^{\beta (\theta-c)(N-3)}}{1-e^{-\beta b}e^{\beta (\theta-c)(N-1)}}=\begin{cases}
1\quad\text{if}\quad \theta\leq c,\\
e^{-2\beta(\theta-c)}\quad\text{if}\quad \theta>c.
\end{cases}
\end{equation}
While for PGG, we have
\begin{equation}
\label{eqn: PGG1}
\frac{1-e^{(N-2)x}}{1-e^{N x}} = \frac{1 - e^{(N-2)\beta[\theta-c(1-\frac{r}{n})-\frac{cr}{n}\frac{n-1}{N-1}]}}{1-e^{N\beta[\theta-c(1-\frac{r}{n})-\frac{cr}{n}\frac{n-1}{N-1}]}}
= \frac{1 - e^{(N-2)\beta(\theta-c(1-\frac{r}{n}))}e^{-\beta(N-2)\frac{cr}{n}\frac{n-1}{N-1}}}{1-e^{N\beta(\theta-c(1-\frac{r}{n}))} e^{-N\beta\frac{cr}{n}\frac{n-1}{N-1}}}
\end{equation}

If $\theta-c(1-\frac{r}{n}) \leq 0$, then 
\begin{align*}
& \lim_{N\rightarrow +\infty} e^{(N-2)\beta(\theta-c(1-\frac{r}{n}))} = \lim_{N\rightarrow +\infty} e^{-\beta(N-2)\frac{cr}{n}\frac{n-1}{N-1}} = 0,
\\& \lim_{N\rightarrow +\infty} e^{-\beta(N-2)c \frac{r(n-1)}{n(N-1)}} = e^{-\beta c \frac{r(n-1)}{n}},
\\&  \lim_{N\rightarrow +\infty} e^{-N\beta c \frac{r(n-1)}{n(N-1)}} = e^{-\beta c \frac{r(n-1)}{n}}.
\end{align*}

Substituting these above limits in \eqref{eqn: PGG1}, it follows that, if $\theta-c(1-\frac{r}{n}) \leq 0$, we have
\begin{equation*}
  \lim\limits_{N\rightarrow +\infty}\frac{1-e^{(N-2)x}}{1-e^{N x}} = 1. 
\end{equation*}

If $\theta-c(1-\frac{r}{n})>0$, then we have

\begin{equation}
\label{eq: PGG2}
   \frac{1 - e^{(N-2)\beta(\theta-c(1-\frac{r}{n}))}e^{-\beta(N-2)\frac{cr}{n}\frac{n-1}{N-1}}}{1-e^{N\beta(\theta-c(1-\frac{r}{n}))} e^{-N\beta\frac{cr}{n}\frac{n-1}{N-1}}}= \frac{\frac{1}{e^{N\beta(\theta-c(1-\frac{r}{n}))}} - e^{-2\beta(\theta-c(1-\frac{r}{n}))} e^{-\beta(N-2)\frac{cr}{n}\frac{n-1}{N-1}}}{\frac{1}{e^{N\beta(\theta-c(1-\frac{r}{n}))}}-e^{-N\beta\frac{cr}{n}\frac{n-1}{N-1}}}.
\end{equation}

We have
\begin{align*}
&    \lim\limits_{N\rightarrow +\infty}\frac{1}{e^{N\beta(\theta-c(1-\frac{r}{n}))}}=0,
\\&\lim\limits_{N\rightarrow+\infty}  e^{-\beta(N-2)\frac{cr}{n}\frac{n-1}{N-1}}=e^{-\beta c \frac{r(n-1)}{n}},
\\&  \lim_{N\rightarrow +\infty} e^{-N\beta c \frac{r(n-1)}{n(N-1)}} = e^{-\beta c \frac{r(n-1)}{n}}.
\end{align*}

From the above limits, \eqref{eqn: PGG1} and \eqref{eq: PGG2} we obtain
\begin{align*}
   \lim\limits_{N\rightarrow +\infty}\frac{1-e^{(N-2)x}}{1-e^{N x}}&= \lim_{N\rightarrow +\infty} \frac{\frac{1}{e^{N\beta(\theta-c(1-\frac{r}{n}))}} - e^{-2\beta(\theta-c(1-\frac{r}{n}))} e^{-\beta(N-2)\frac{cr}{n}\frac{n-1}{N-1}}}{\frac{1}{e^{N\beta(\theta-c(1-\frac{r}{n}))}}-e^{-N\beta\frac{cr}{n}\frac{n-1}{N-1}}}
   \\&=\frac{- e^{-2\beta(\theta-c(1-\frac{r}{n}))} e^{-\beta\frac{cr(n-1)}{n}}}{-e^{-\beta\frac{cr(n-1)}{n}}}
   \\&= e^{-2\beta(\theta-c(1-\frac{r}{n}))}.
\end{align*}

Hence for PGG
\begin{equation}
\label{eq: PGG3}
    \lim\limits_{N\rightarrow +\infty}\frac{1-e^{(N-2)x}}{1-e^{N x}}=\begin{cases}
    1,\quad \text{if}\quad \theta\leq c(1-\frac{r}{n}),\\
    e^{-2\beta(\theta-c(1-\frac{r}{n}))}\quad\text{if}\quad \theta> c(1-\frac{r}{n}).
    \end{cases}
\end{equation}

From \eqref{eqn: omega}, \eqref{eqn: exponential}, and \eqref{eq: DG1}, we obtain, for DG 
\begin{align}
\label{eq: es2}
\lim_{N\rightarrow+\infty}w(\theta)&=\begin{cases}
1+e^{\beta(\theta-c)}\quad\text{if}\quad \theta \leq c,\\
1+e^{-\beta(\theta-c)}\quad\text{if}\quad \theta>c
\end{cases}\notag\\
&=1+e^{-\beta |\theta-c|},
\end{align}

and similarly, from \eqref{eqn: omega}, \eqref{eqn: exponential}, and \eqref{eq: PGG3}, we get, for PGG
\begin{align}
\label{eq: es21}
\lim_{N\rightarrow+\infty}w(\theta)&=\begin{cases}
1+e^{\beta(\theta-c)}e^{\beta c\frac{r}{n}}\quad\text{if}\quad \theta \leq c(1-\frac{r}{n}),\\
1+e^{-\beta(\theta-c)}e^{-\beta c\frac{r}{n}}\quad\text{if}\quad \theta>c(1-\frac{r}{n})
\end{cases}\nonumber
\\&=1+ e^{-\beta |\theta-c(1-\frac{r}{n})|}.
\end{align}

Now we study the limit of $\frac{H_{N,a,b}+\underline{h}}{H_{a,b}}$ and of $\frac{H_{N,a,b}+\overline{h}}{H_{a,b}}$ as $N\rightarrow+\infty$.

For $i=1,\ldots, N_{a,b}$, $h_i=\frac{1}{a(N-i)}< \frac{1}{a(N-N_{a,b})}$ as $N-i\geq N-N_{a,b}$. Since $\lim\limits_{N\rightarrow +\infty}\frac{1}{a(N-N_{a,b})}=0$, it follows that $\lim\limits_{N\rightarrow +\infty} h_i=0$.

For $i=N_{a,b}+1,\ldots, N-2$, $h_i=\frac{1}{b(i+1)}< \frac{1}{b(N_{a,b}+1)}$ as $i\geq N_{a,b}+1$. Since $\lim\limits_{N\rightarrow +\infty}\frac{1}{b(N_{a,b}+2)}=0$, it follows that $\lim\limits_{N\rightarrow +\infty} h_i=0$.

As $N_{a,b}=\lfloor \frac{N}{\frac{b}{a}+1} \rfloor$ $\rightarrow +\infty$ for $N\rightarrow +\infty$, we deduce that $h_j \rightarrow 0$ as $N\rightarrow +\infty$ for any $j$,
\begin{equation}
\label{eq: es3}
\lim_{N\rightarrow +\infty}\frac{H_{N,a,b}}{H_{a,b}}+\frac{\underline{h}}{H_{a,b}}=\lim_{N\rightarrow +\infty}\frac{H_{N,a,b}}{H_{a,b}}+\frac{\overline{h}}{H_{a,b}}=1.
\end{equation}

Therefore, from  \eqref{eq: es0}, \eqref{eq: es2},  and \eqref{eq: es3} we obtain, for DG: 
\begin{equation*}
\lim_{N\rightarrow +\infty}\frac{E_{mix}(\theta)}{\frac{N^2\theta}{2}H_{N,a,b}}=1+e^{-\beta|\theta-c|},
\end{equation*}

and from \eqref{eq: es0}, \eqref{eq: es21} and \eqref{eq: es3}, for PGG:
\begin{equation*}
\lim_{N\rightarrow +\infty}\frac{E_{mix}(\theta)}{\frac{N^2\theta}{2}H_{N,a,b}}=1+ e^{-\beta |\theta-c(1-\frac{r}{n})|}.
\end{equation*}
\end{proof}

\begin{figure}[H]
\centering
\includegraphics[width=1\textwidth]{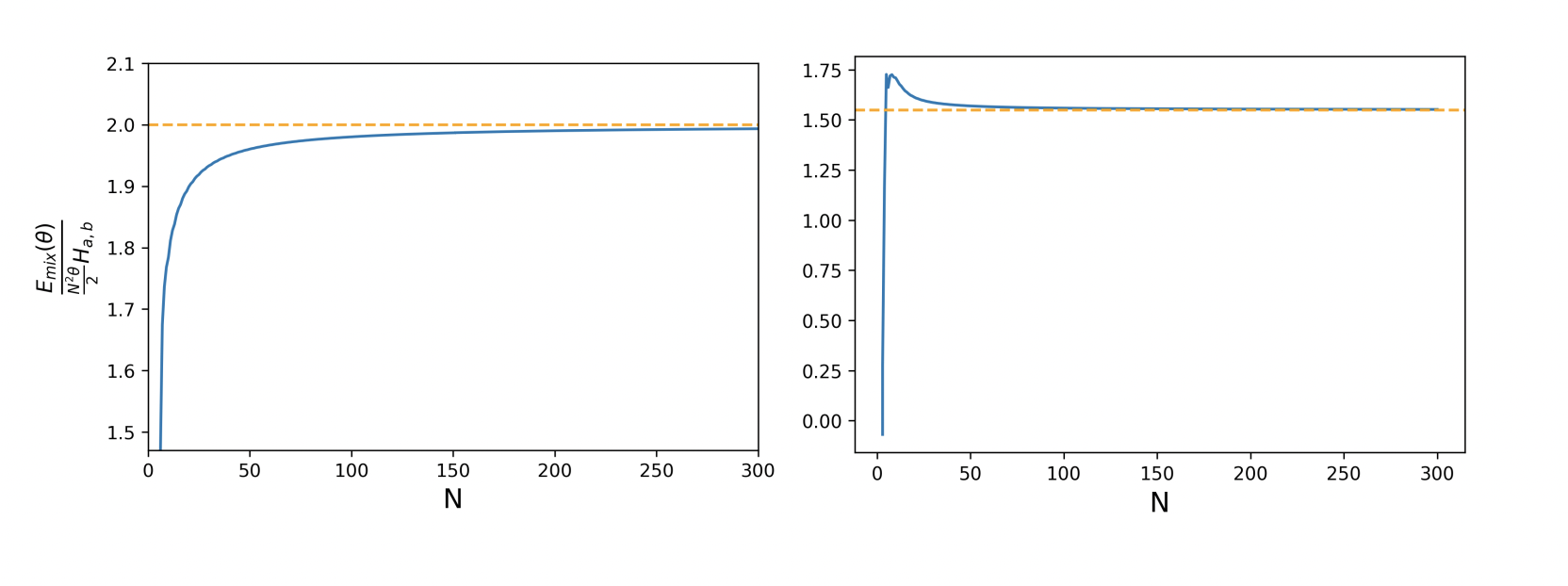}
\caption{Large population size limit. We calculate numerically the expected total cost for the mixed incentive, varying the population size $N$. The dashed orange lines represent the corresponding theoretical limiting values obtained in Proposition \ref{prop: large N limit reward} for the large population size limit, $N\rightarrow +\infty$. We observe that numerical results are in close accordance with those obtained theoretically. Results are obtained for DG with $N = 2500, B = 2, c = 1, \beta = 1, \theta = 1$ (left) and for PGG with $N = 2500, c = 1, r = 3, n = 5, \beta = 1, \theta = 1$ (right).}
\label{fig:infinitepopulation}
\end{figure}

\section{Discussion}
\label{sec: discussion}
The use of institutional incentives such as reward and punishment is an effective tool for the promotion of cooperation in social dilemmas, as proven both by theoretical  \cite{hauert2007,carrotstick,sigmundinstitutions,han2022institutional,duong2021cost}  and experimental results \cite{ostrom2009understanding,van2014reward,wu2014role}. In past works, although mixed incentives were used, the aspect of minimisation of the cost function while ensuring a minimum level of cooperation was overlooked. Moreover, existing works that consider this question usually omit the stochastic effects that drive population dynamics, namely when the strength of selection, $\beta$, varies.

In this work, we used a stochastic evolutionary game theoretic approach for finite, well-mixed populations and obtained theoretical results for the optimal cost of mixed incentives that ensure a desired level of cooperation, for a given intensity of selection, $\beta$. We show that this cost strongly depends on the value of $\beta$ due to the existence of a phase transition in the cost function for providing mixed incentives. This behaviour is missing in works that consider a deterministic evolutionary approach \cite{wang2019exploring}. We also characterised asymptotic behaviours of the cost function and showed that mixed incentives are always cost efficient than either using only reward or only punishment. In particular, we successfully  obtained an infinite population limit, as well as those for neutral drift and strong selection.

For the mathematical analysis of the mixed incentive cost function to be possible, we made some assumptions. Firstly, in order to derive the analytical formula for the frequency of cooperation, we assumed a small mutation limit \cite{rockenbach,nowak2004emergence,sigmund2010calculus}. Despite the simplified assumption, this small mutation limit approach has  wide applicability to scenarios which go well beyond the strict limit of very small mutation rates \cite{zisis2015generosity,hauert2007,sigmundinstitutions,rand2013evolution}. If we were to relax this assumption, the derivation of a closed form for the frequency of cooperation would be intractable. Secondly, we focused on two important cooperation dilemmas, the Donation Game and the Public Goods Game. Both have in common that the difference in average payoffs between a cooperator and a defector does not depend on the population composition. This special property allowed us to simplify the fundamental matrix of the Markov chain to a tridiagonal form and apply the techniques of matrix analysis to obtain a closed form of its inverse matrix. In games with more complex payoff matrices such as the general prisoners' dilemma and the collective risk game \cite{sun2021combination}, this property no longer holds (e.g. in the former case the payoff difference, $\Pi_C(i)-\Pi_D(i)$, depends additively on $i$) and the technique in this paper cannot be directly applied.  In these scenarios, we might consider other approaches to approximate the inverse matrix, exploiting its block structure.

More recent works looked at the effect of indirect exclusion on cooperation, having found that the introduction of indirect exclusion could induce the stable coexistence of cooperators and defectors or the dominance of cooperators, successfully promoting cooperation \cite{Liu_2022}. Looking at prior agreement before an interaction takes place also showed that cooperation would be increased. Thus, individuals choose whether to take part in a social dilemma and those who do are rewarded \cite{ogbo2022evolution,han2022institutional}. Our future work will consider cost-efficiency in these different form of institutional incentives, including the problem of providing incentives whenever the frequency or number of cooperators (or defectors) in the population  does not exceed a given threshold.

We furthermore aim to investigate the optimisation problems of incentives such as reward, punishment, and exclusion in complex networks. There has been little attention to providing analytical results for cost-efficient incentives in structured populations or in more complex games, so this would also be an interesting research avenue. Finally, since time is most precious, we intent to explore the time that the system needs in order to achieve a desired level of cooperation.

\section{Appendix}
\label{sec: appendix}
In this appendix, we provide explicit calculations for some small population cases, as well as detailed proofs and computations of technical results in the previous sections.
\subsection{Small population examples}

\subsubsection{$N=3$}
The cost function is
\begin{align*}
E_{mix}(\theta) &= \frac{\theta}{2}\sum_{j=1}^{2}(n_{1j}+n_{2,j}) \min\Big(\frac{j}{a}, \frac{3-j}{b}\Big)\\
&=\frac{\theta}{2} \sum_{j=1}^{2} \frac{3^2}{(3-j)j} \Big((W^{-1})_{1j} +(W^{-1})_{2,j}\Big) \min(j,3-j) \notag\\
&= \frac{9 \theta(1+e^x)}{2(e^{3x}-1)} \sum_{j=1}^{2} \frac{e^{2x} - e^{(j-1)x} + e^{jx} - 1 }{j(3-j)}\min(j, 3-j)\notag
\\&=\frac{27 \theta  (1 + e^x)^2}{4(e^{x} +e^{2x} +1 )}.
\end{align*}
The function $F$ is
$$
F(u) = \frac{u^4 + 3u^3 + 4u^2 + 3u +1}{u^3 - u} - \log(u).
$$
The critical threshold is
$$
\beta^*= 3.67.
$$
\subsubsection{$N=4$}
The cost function is
\begin{align*}
E_{mix}(\theta) &= \frac{\theta}{2}\sum_{j=1}^{3}(n_{1j}+n_{3,j}) \min(\frac{j}{a}, \frac{4-j}{b})\\
&=\frac{\theta}{2} \sum_{j=1}^{3} \frac{4^2}{(4-j)j} \Big((W^{-1})_{1j} +(W^{-1})_{3,j}\Big) \min(j, 4-j) \notag\\
&= \frac{8 \theta(1+e^x)}{(e^{4x}-1)} \sum_{j=1}^{3} \frac{e^{3x} - e^{(j-1)x} + e^{jx} - 1 }{j(4-j)}\min(j, 4-j)\notag
\\&=\frac{4 \theta (1 + e^x)(9+ 10 e^x + 10 e^{2x})}{3 (1 + e^x + e^{2x} + e^{3x})}.
\end{align*}
The function $F$ is
$$
F(u) = \frac{1.16u^6 + 3.66u^5 + 6.16u^4 + 7.33u^3 + 6.16u^2 + 3.66u + 1.16}{1.33u^5 + 2.66u^4 - 2.66u^2 - 1.33u} - \log(u).
$$
The critical threshold is 
$$
\beta^*= 2.34.
$$
\subsubsection{$N=5$}
The cost function is
\begin{align*}
E_{mix}(\theta) &= \frac{\theta}{2}\sum_{j=1}^{4}(n_{1j}+n_{4,j}) \min\Big(\frac{j}{a}, \frac{5-j}{b}\Big)\\
&=\frac{\theta}{2} \sum_{j=1}^{4} \frac{5^2}{(5-j)j} \Big((W^{-1})_{1j} +(W^{-1})_{4,j}\Big) \min(j, 5-j) \notag\\
&= \frac{25\theta(1+e^x)}{2(e^{5x}-1)} \sum_{j=1}^{4} \frac{e^{4x} - e^{(j-1)x} + e^{jx} - 1 }{j(5-j)}\min(j, 5-j)\notag
\\&=\frac{25\theta(1+e^x)(17+18e^{x}+18e^{2x}+17e^{3x})}{24(1+e^x+e^{2x}+e^{3x}+e^{4x})}.
\end{align*}
The function $F$ is
$$
F(u) = \frac{1.16u^8 + 3.58u^7 + 6.083u^6 + 8.5u^5 + 9.66u^4 + 8.5u^3 + 6.083u^2 + 3.58u + 1.16}{1.25u^7 + 2.66u^6 + 3.83u^5 + 8.88u^4 - 3.83u^3 - 2.66u^2 - 1.25u} - \log(u).
$$
The critical threshold is 
$$
\beta^*= 2.084.
$$

\subsection{Proof of Proposition \ref{prop: 1st result}}
\label{sec: proof of 1st result}
In this section, we prove Proposition \ref{prop: 1st result}. 
\begin{proof}
We recall that:
$$
E_{r}(\theta) = \frac{\theta}{2}\sum_{j=1}^{N-1}(n_{1j}+n_{N-1,j})\cdot \frac{j}{a},\quad E_{p}(\theta)=\frac{\theta}{2}\sum_{j=1}^{N-1}(n_{1j}+n_{N-1,j})\cdot \frac{N-j}{b},
$$ 
and 
$$
E_{mix}(\theta)= \frac{\theta}{2}\sum_{j=1}^{N-1}(n_{1j}+n_{N-1,j})\cdot \min\Big(\frac{j}{a}, \frac{N-j}{b}\Big).
$$ 

The first statement follows directly from these formulae, the fact that $n_{1,j}, n_{N-1,j}>0$, and the elementary inequalities
$$
\min\Big(\frac{j}{a}, \frac{N-j}{b}\Big)\leq \frac{j}{a} \quad \text{and}\quad \min\Big(\frac{j}{a}, \frac{N-j}{b}\Big)\leq \frac{N-j}{b}.
$$

Now we prove the second statement. The third one is similar. Suppose that $\frac{b}{a}\leq \frac{1}{N-1}$. We have 
\begin{align*}
\frac{b}{a}\leq \frac{1}{N-1}=\frac{N}{N-1}-1= \min_{1\leq j \leq N-1} \Big(\frac{N-j}{j}\Big),
\end{align*}
which implies that 
$$
 \frac{b}{a}\leq \frac{N-j}{j}, \quad \text{i.e.}, \quad \frac{j}{a}\leq \frac{N-j}{b} \quad\text{for all} \quad j=1,\ldots, N-1.
$$

Thus, $\min\Big(\frac{j}{a},\frac{N-j}{b}\Big)=\frac{j}{a}$ for all $j=1,2,\ldots,N-1$, where $N$ is the number of individuals in the population. Therefore, 
\begin{align*}
E_{mix}(\theta)&=\frac{\theta}{2}\sum_{j=1}^{N-1}(n_{1j}+n_{N-1,j})\min\Big(\frac{j}{a}, \frac{N-j}{b}\Big)
\\&=\frac{\theta}{2}\sum_{j=1}^{N-1}(n_{1j}+n_{N-1,j})\frac{j}{a}
\\&=E_r(\theta).
\end{align*}
\end{proof}

\subsection{Proof of Lemma \ref{lem: Emix and derivative}}
\label{sec: prooflem2}

In this section, we give the proof of Lemma \ref{lem: Emix and derivative}.

\begin{proof}
Since $x=\beta(\theta+\delta)$, we have
\begin{align}
\label{eq: derivative E}
\frac{d}{d\theta}E_{mix}(\theta)=\beta\frac{d}{dx}E_{mix}(\theta)&=\frac{\beta N^2}{2}\Bigg\{\frac{1}{\beta}\frac{f(x)}{g(x)}+\theta\Big(\frac{f'(x)}{g(x)}-\frac{f(x)g'(x)}{g(x)^2}\Big)\Bigg\}\notag
\\&=\frac{N^2}{2}\Bigg\{\left[\frac{f(x)}{g(x)}+(x-\beta \delta)\Big(\frac{f'(x)}{g(x)}-\frac{f(x)g'(x)}{g(x)^2}\Big)\right]\Bigg\}\notag
\\&=\frac{N^2}{2}\frac{1}{g(x)^2}\Big[f(x)g(x)-(x-\beta \delta)(f(x)g'(x)-f'(x)g(x))\Big].
\end{align}
Let $u:=e^x$. From the formula of $f(x)=(1+e^x)\Big[\big(1+e^x+\ldots +e^{(N-2)x}\big)H_{N,a,b} + \sum_{j=1}^{N-1}\frac{e^{(j-1) x}}{j(N-j)}\min\Big(\frac{j}{a}, \frac{N-j}{b}\Big)\Big]$ and $g(x)=1+e^x+\ldots+e^{(N-1)x}$, it is more convenient to express the right-hand side of \eqref{eq: derivative E} in terms of $u$. We have
\begin{align*}
f(x)&=(1+u)\sum_{j=0}^{N-2} u^j\Big(H_{N,a,b} + \frac{\min(\frac{j+1}{a},\frac{N-j-1}{b})}{(j+1)(N-j-1)}\Big),
\\f'(x)&=\sum_{j=0}^{N-2}\Big(H_{N,a,b} + \frac{\min(\frac{j+1}{a},\frac{N-j-1}{b})}{(j+1)(N-j-1)}\Big)u^j(j(1+u)+u),
\\g(x)&=\sum_{j=0}^{N-1} u^j, \quad g'(x)=\sum_{j=1}^{N-1} j u^{j}.
\end{align*}
Therefore
\begin{align*}
f(x)g'(x)-f'(x)g(x)&=(1+u)\Big(\sum_{j=0}^{N-2}(H_{N,a,b} + \frac{\min(\frac{j+1}{a},\frac{N-j-1}{b})}{(j+1)(N-j-1)}\Big)\Big(\sum_{j=1}^{N-1} j u^{j}\Big)
\\&\quad-\Big(\sum_{j=0}^{N-2}\Big(H_{N,a,b} + \frac{\min(\frac{i+1}{a},\frac{N-i-1}{b})}{(i+1)(N-i-1)}\Big) u^j(j(1+u)+u)\Big)\Big(\sum_{j=0}^{N-1}u^j\Big)
\\&\quad=(1+u)u\Bigg[\Big(\sum_{j=0}^{N-2}\Big(H_{N,a,b} + \frac{\min(\frac{j+1}{a},\frac{N-j-1}{b})}{(j+1)(N-j-1)}\Big) u^j\Big)\Big(\sum_{j=1}^{N-1} j u^{j-1}\Big)
\\&\qquad\qquad\qquad-\Big(\sum_{j=1}^{N-2}\Big(H_{N,a,b} + \frac{\min(\frac{j+1}{a},\frac{N-j-1}{b})}{(j+1)(N-j-1)}\Big) j u^{j-1}\Big)\Big(\sum_{j=0}^{N-1}u^j\Big)\Bigg]
\\&\qquad\qquad\qquad- u\Big(\sum_{j=0}^{N-2}\Big(H_{N,a,b} + \frac{\min(\frac{j+1}{a},\frac{N-j-1}{b})}{(j+1)(N-j-1)}\Big) u^j\Big)\Big(\sum_{j=0}^{N-1}u^j\Big)
\\&\quad:=uP(u),
\end{align*}
where we define
\begin{align}
\label{eq: def of P}
P(u)&:=(1+u)\Bigg[\Big(\sum_{j=0}^{N-2}\Big(H_{N,a,b} + \frac{\min(\frac{j+1}{a},\frac{N-j-1}{b})}{(j+1)(N-j-1)}\Big) u^j\Big)\Big(\sum_{j=1}^{N-1} j u^{j-1}\Big)\notag
\\&\qquad\qquad-\Big(\sum_{j=1}^{N-2}\Big(H_{N,a,b} + \frac{\min(\frac{j+1}{a},\frac{N-j-1}{b})}{(j+1)(N-j-1)}\Big) j u^{j-1}\Big)\Big(\sum_{j=0}^{N-1}u^j\Big)\Bigg]\notag
\\&\qquad\qquad-\Big(\sum_{j=0}^{N-2}\Big(H_{N,a,b} + \frac{\min(\frac{j+1}{a},\frac{N-j-1}{b})}{(j+1)(N-j-1)}\Big) u^j\Big)\Big(\sum_{j=0}^{N-1}u^j\Big).
\end{align}
We also have
\begin{align}
\label{eq: def of Q}
    f(x)g(x)&=(1+u)\Big(\sum_{j=0}^{N-2}\Big(H_{N,a,b} + \frac{\min(\frac{j+1}{a},\frac{N-j-1}{b})}{(j+1)(N-j-1)}\Big) u^j\Big)\Big(\sum_{j=0}^{N-1} u^{j}\Big):=Q(u).
\end{align}

Note that $x-\beta\delta=\beta \theta>0$ and$f(x)g(x)>0$. It follows that if $f(x)g'(x)-f'(x)g(x)\leq 0$ then $E_{mix}'(\theta)>0$. Therefore, we only need to consider the case $f(x)g'(x)-f'(x)g(x)=u P(u)> 0$ (i.e., when $P(u)>0$). Let $u=e^x$ and define
\begin{equation}
\label{eq: F}
F(u):=\frac{Q(u)}{uP(u)}-\log (u).
\end{equation}
Then
\begin{equation}
\label{eq: derivative Emix}
E_{mix}'(\theta)=\frac{N^2}{2g(x)^2}(uP(u))\Big(\frac{Q(u)}{uP(u)}-\log(u) + \beta \delta\Big)=\frac{N^2}{2g(x)^2}(uP(u))(F(u) + \beta \delta).    
\end{equation}

The next step is to understand the sign of the term $F(u)+\beta \delta$, which is required to understand the polynomials $P$ and $Q$. In the next section, we analyse the polynomial $P$ ($Q$ is explicit and is much simpler).  
\end{proof}

\subsection{Proof of Proposition \ref{prop: coeffs}}
\label{sec:proof1}
In this section, we provide a proof of Proposition \ref{prop: coeffs}. We delicately analyse the coefficients of $P$ and employ the discrete integration by parts techniques from numerical analysis.
\begin{proof} 
Let $m_j=\frac{\min(\frac{j+1}{a},\frac{N-j-1}{b})}{(j+1)(N-j-1)}$ and let $a_j:=H_{N,a,b}+m_j$ for $j=0,\ldots, N-2$. 
Suppose that $P(u)=\sum_{j=0}^{2N-4}p_j u^j$. Using the following formula of product of two polynomials:
$$
\Big(\sum_{j=0}^m a_j x^j\Big)\Big(\sum_{j=0}^n b_j x^j\Big)=\sum_{k=0}^{m+n}\Big(\sum_{j=\max\{0,k-n\}}^{\min\{m,k\}} a_j b_{k-j}\Big) x^k,
$$

we have
\begin{align*}
(i)~~\Big(\sum_{j=0}^{N-2}(H_{N,a,b}+m_j) u^j\Big)\Big(\sum_{j=1}^{N-1} j u^{j-1}\Big)&=\Big(\sum_{j=0}^{N-2}a_j u^j\Big)\Big(\sum_{j=1}^{N-2} (j+1) u^{j}\Big)
\\&=\sum_{k=0}^{2N-4}\Big(\sum_{j=\max(k-N+2,0)}^{\min(k,N-2)} a_j (k-j+1)\Big) u^k\\
(ii)~~\Big(\sum_{j=1}^{N-2}\Big(H_{N,a,b} + m_j\Big) j u^{j-1}\Big)\Big(\sum_{j=0}^{N-1}u^j\Big)&=\Big(\sum_{j=0}^{N-3}a_{j+1}(j+1) u^{j}\Big)\Big(\sum_{j=0}^{N-1}u^j\Big)
\\&=\sum_{k=0}^{2N-4}\Big(\sum_{j=\max(0,k-N+1)}^{\min(k,N-3)}a_{j+1}(j+1)\Big)u^k.\\
(iii)~~ \Big(\sum_{j=0}^{N-2}\Big(H_{N,a,b} + m_j\Big) u^j\Big)\Big(\sum_{j=0}^{N-1}u^j\Big)&=\Big(\sum_{j=0}^{N-2}a_j u^j\Big)\Big(\sum_{j=0}^{N-1}u^j\Big)
\\&=\sum_{k=0}^{2N-3}\Big(\sum_{j=\max(0,k-N+1)}^{\min(N-2,k)}a_j\Big)u^k.
\end{align*}

Hence 
\begin{align*}
(iii)~~&\Bigg[\Big(\sum_{j=0}^{N-2}(H_{N,a,b}+m_j) u^j\Big)\Big(\sum_{j=1}^{N-1} j u^{j-1}\Big)-\Big(\sum_{j=1}^{N-2}\Big(H_{N,a,b} + m_j\Big) j u^{j-1}\Big)\Big(\sum_{j=0}^{N-1}u^j\Big)\Bigg](1+u)
\\ \quad&=\Bigg(\sum_{k=0}^{2N-4}\Big(\sum_{j=\max(k-N+2,0)}^{\min(k,N-2)} a_j (k-j+1)-\sum_{j=\max(0,k-N+1)}^{\min(k,N-3)}a_{j+1}(j+1)\Big) u^k\Bigg)(1+u)\\
\\&=\Big(\sum_{k=0}^{2N-4}b_k u^k\Big)(1+u)
\\&=\sum_{k=0}^{2N-3}\Big(\sum_{j=\max(0,k-1)}^{\min(k,2N-4)}b_j\Big)u^k
\\&=b_0+\sum_{k=1}^{2N-4}(b_{k-1}+b_{k})u^k+b_{2N-4}u^{2N-3}
\end{align*}
where for $k=0,\ldots, 2N-4$
\begin{align}
b_k&:=\sum_{j=\max(k-N+2,0)}^{\min(k,N-2)} a_j (k-j+1)-\sum_{j=\max(0,k-N+1)}^{\min(k,N-3)}a_{j+1}(j+1)
\\&=\sum_{j=\max(k-N+3,1)}^{\min(k+1,N-1)} a_{k-j+1} j-\sum_{j=\max(1,k-N+2)}^{\min(k+1,N-2)}a_j j\notag
\\&=\begin{cases}
\sum_{j=1}^{k+1}(a_{k-j+1}-a_j)j\qquad k\leq N-3,\\
\sum_{j=1}^{N-1}a_{k-j+1}j -\sum_{j=1}^{N-2}a_j j\qquad k=N-2,\\
\sum_{j=k-N+3}^{N-1}a_{k-j+1} j-\sum_{j=k-N+2}^{N-2}a_j j\qquad N-1\leq k\leq 2N-4.
\end{cases}\notag
\end{align}

Therefore, we obtain
\begin{align*}
P(u)&=\sum_{k=0}^{2N-3}\Big(\sum_{j=\max(0,k-1)}^{\min(k,2N-4)}b_j-\sum_{j=\max(0,k-N+1)}^{\min(N-2,k)}a_j\Big)u^k,
\end{align*}
that is for $k=0,\ldots, 2N-3$
\begin{align}
p_k&=\sum_{j=\max(0,k-1)}^{\min(k,2N-4)}b_j-\sum_{j=\max(0,k-N+1)}^{\min(N-2,k)}a_j
\\&=\begin{cases}
b_0-a_0\qquad k=0,\\
b_{k-1}+b_k-\sum_{j=0}^{k}a_{j}\qquad 1\leq k\leq N-2,\\
b_{k-1}+b_k-\sum_{j=k-N+1}^{N-2}a_j\qquad N-1\leq k\leq 2N-4,\\
b_{2N-4}-a_{N-2}\qquad k=2N-3.
\end{cases}
\end{align}

In particular,
\begin{align*}
p_0 &=b_0-a_0=a_0 - a_1 - a_0=-a_1 
\\&= - \Big[H_{N,a,b} + \frac{\min(\frac{2}{a},\frac{N-2}{b})}{2(N-2)}\Big] <0
\\ p_{2N-3}&=b_{2N-4}-a_{N-2}= (N-1)a_{N-2} - (N-2)a_{N-2} - a_{N-2}=0,
\\ p_{2N-4} &= b_{2N-5} + b_{2N-4} - a_{N-3}  - a_{N-2} 
\\&= \Big[(N-2)a_{N-2} + (N-1)a_{N-3} - (N-3)a_{N-3} - (N-2)a_{N-2}\Big] + \Big[(N-1)a_{N-2} - (N-2)a_{N-2}\Big] 
\\&\qquad- a_{N-3}  - a_{N-2} 
    \\&= a_{N-3}
    \\&=\Big[H_{N,a,b} + \frac{\min(\frac{N-2}{a},\frac{2}{b})}{2(N-2)} \Big] > 0
\end{align*}

Hence, in particular, if $a=b$ then
$$
p_0+p_{2N-4}=0.
$$

Now we will show that when $a=b$, $p_{N-2}=0$ and $p_{k}+p_{2N-4-k}=0$ for all $1\leq k\leq N-3$.

\textbf{For $1\leq k\leq N-3$:}
\begin{align*}
b_{k-1}+b_{k}&=\sum_{j=1}^k(a_{k-j}-a_j)j+\sum_{j=1}^{k+1}(a_{k-j+1}-a_j)j
\\&=\sum_{j=1}^k(a_{k-j}-2a_j+a_{k-j+1})j+(a_0-a_{k+1})(k+1)
\\&=\sum_{j=1}^k (m_{k-j+1} + m_{k-j} - 2m_j)j +(m_0-m_{k+1})(k+1).
\end{align*}

In addition,
\begin{align*}
       \sum_{j=0}^{k}a_{j}&=\sum_{j=0}^{k}\Big(H_{N,a,b}+ m_{j}\Big)
    \\&=(k+1)H_{N,a,b}+\sum_{j=0}^{k}m_j.
\end{align*}

Therefore, for $1\leq k\leq N-3$, we have
\begin{align}
p_k&=b_{k-1}+b_k-\sum_{j=0}^{k}a_j\notag
\\&=\sum_{j=1}^k (m_{k-j+1} + m_{k-j} - 2m_j)j +(m_0-m_{k+1})(k+1)-(k+1)H_{N,a,b}-\sum_{j=0}^{k}m_j.\label{eq:pk}
\end{align}

\textbf{For $k=N-2$}, we have
\begin{align*}
    \sum_{j=0}^{N-2}a_{j}&=\sum_{j=0}^{N-2}\Big(H_{N,a,b}+ m_{j}\Big)
    \\&=(N-1)H_{N,a,b}+\sum_{j=0}^{N-2}\frac{\min(\frac{j+1}{a},\frac{N-j-1}{b})}{(j+1)(N-j-1)}
    \\&=(N-1)H_{N,a,b}+\sum_{j=1}^{N-1}\frac{\min(\frac{j}{a},\frac{N-j}{b})}{j(N-j)}
    \\&=(N-1)H_{N,a,b}+H_{N,a,b}
    \\&=N H_{N,a,b}.
\end{align*}

It follows from the above computations that 
$$
\sum_{j=0}^{N-2}m_j=H_{N,a,b}.
$$

Hence
\begin{align*}
p_{N-2}&=b_{k-1}+b_k-\sum_{j=0}^{k}a_{j}
\\&=b_{N-3}+b_{N-2} -\sum_{j=0}^{N-2}a_{j}
\\&=\sum_{j=1}^{N-2}(a_{N-2-j}-a_j)j+\sum_{j=1}^{N-1}a_{N-1-j}j -\sum_{j=1}^{N-2}a_j j - \sum_{j=0}^{N-2}a_{j}
\\&=\sum_{j=1}^{N-2}(a_{N-2-j}-2a_j+a_{N-1-j})j + (N-1)a_0 - \sum_{j=0}^{N-2}a_{j}
\\&=\sum_{j=1}^{N-2}(H_{N,a,b} + m_{N-2-j} - 2H_{N,a,b} - 2m_{j} + H_{N,a,b} + m_{N-1-j})j + (N-1)a_0-N H_{N,a,b}
\\&= \sum_{j=1}^{N-2}( m_{N-2-j} - 2m_{j} +  m_{N-1-j})j + (N-1)a_0-N H_{N,a,b}
\end{align*}

We have 
$$
H_{N,a,b}=\sum\limits_{j=1}^{N-1}\frac{1}{j(N-j)}\min(\frac{j}{a},\frac{N-j}{b})=\sum\limits_{j=1}^{N-1}\frac{1}{j(N-j)}\min(\frac{N-j}{a},\frac{j}{b})
$$
$$
m_j=\frac{\min(\frac{j+1}{a},\frac{N-j-1}{b})}{(j+1)(N-j-1)},\quad m_{N-2-j}=\frac{\min((N-1-j)/a, (j+1)/b)}{(j+1)(N-j-1)}, \quad m_{N-1-j}=\frac{\min((N-j)/a, j/b)}{j(N-j)}
$$

Hence 

\begin{align*}
&\sum_{j=1}^{N-2}m_{N-2-j} j =\sum_{j=1}^{N-2}\frac{\min((N-1-j)/a, (j+1)/b)}{(j+1)(N-j-1)}j=\sum_{j=2}^{N-1}\frac{\min((N-j)/a, j/b)}{j(N-j)}(j-1)
\\&=\sum_{j=2}^{N-1}\frac{\min((N-j)/a, j/b)}{(N-j)}-\sum_{j=2}^{N-1}\frac{\min((N-j)/a, j/b)}{(N-j)}
\\&=\sum_{j=2}^{N-1}\frac{\min((N-j)/a, j/b)}{(N-j)}-\Big[H_{N,a,b}-\frac{\min((N-1)/a, 1/b)}{(N-1)}\Big],\\
&\sum_{j=1}^{N-2}m_{N-j-1} j=\sum_{j=1}^{N-2}\frac{\min((N-j)/a, j/b)}{j(N-j)}j=\sum_{j=1}^{N-2}\frac{\min((N-j)/a, j/b)}{(N-j)},\\
& \sum_{j=1}^{N-2}m_j j=\sum_{j=1}^{N-2}\frac{\min(\frac{j+1}{a},\frac{N-j-1}{b})}{(j+1)(N-j-1)}j=\sum_{j=2}^{N-1}\frac{\min(\frac{j}{a},\frac{N-j}{b})}{j(N-j)}(j-1)
\\&=\sum_{j=2}^{N-1}\frac{\min(\frac{j}{a},\frac{N-j}{b})}{(N-j)}-\Big[H_{N,a,b}-\frac{\min((N-1)/a, 1/b)}{(N-1)}\Big].
\end{align*}

Therefore, recalling that
\begin{align*}
a_0&= H_{N,a,b}+m_0=H_{N,a,b}+\frac{\min(1/a, (N-1)/b)}{N-1}, \\
H_{N,a,b}&=\sum_{j=2}^{N-1}\frac{\min((N-j)/a, j/b)}{j(N-j)}
\\&=\frac{1}{N}\sum_{j=2}^{N-1}\min((N-j)/a, j/b)\Big[\frac{1}{j}+\frac{1}{N-j}\Big]
\\&=\frac{1}{N}\Big[\sum_{j=2}^{N-1}\frac{\min((N-j)/a, j/b)}{j}+\sum_{j=2}^{N-1}\frac{\min((N-j)/a, j/b)}{N-j}\Big]
\end{align*}

Hence,
\begin{align*}
p_{N-2} &=\sum_{j=2}^{N-1}\frac{\min((N-j)/a, j/b)}{N-j}  - 2\sum_{j=2}^{N-1}\frac{\min(\frac{j}{a},\frac{N-j}{b})}{N-j} + \sum_{j=1}^{N-2}\frac{\min((N-j)/a, j/b)}{N-j} 
\\&\qquad+\Big[H_{N,a,b}-\frac{\min((N-1)/a, 1/b)}{(N-1)}\Big]+  (N-1)H_{N,a,b}+\min(1/a, (N-1)/b)-N H_{N,a,b}
\\&=\sum_{j=2}^{N-1}\frac{\min((N-j)/a, j/b)}{N-j}  - 2\sum_{j=2}^{N-1}\frac{\min(\frac{j}{a},\frac{N-j}{b})}{N-j} + \sum_{j=2}^{N-1}\frac{\min((N-j)/a, j/b)}{N-j}
\\&=2\sum_{j=2}^{N-1}\frac{\min((N-j)/a, j/b)-\min(\frac{j}{a},\frac{N-j}{b})}{N-j}.
\end{align*}

If $a=b$, then
$$
\min((N-j)/a, j/b)-\min(\frac{j}{a},\frac{N-j}{b})=\frac{1}{a}\Big[\min(N-j, j)-\min(N-j, j)\Big]=0,
$$

hence, in this case, we have $p_{N-2}=0$.

\textbf{For $N-1\leq k\leq 2N-5$:}
\begin{align*}
b_{k-1}+b_k&=\sum_{j=k-N+2}^{N-2}(a_{k-j}-a_j)j+(2N-2-k)a_{k-N+1}+\sum_{j=k-N+3}^{N-2}(a_{k-j+1}-a_j)j+(2N-3-k)a_{k-N+2}
\\&=\sum_{j=k-N+3}^{N-2}(a_{k-j+1}-2a_j+a_{k-j})j+(k-N+2)(a_{N-2}-a_{k-N+2})
\\&\qquad+(2N-2-k)a_{k-N+1}+(2N-3-k)a_{k-N+2}
\\&=\sum_{j=k-N+3}^{N-2}(m_{k-j}-2m_j+m_{k-j+1})j 
\\&\qquad+(k-N+2)(m_{N-2}-m_{k-N+2})+(2N-2-k)(H_{N,a,b}+m_{k-N+1})+(2N-3-k)(H_{N,a,b}+m_{k-N+2})
\\&=\sum_{j=k-N+3}^{N-2}(m_{k-j}-2m_j+m_{k-j+1})j 
\\&\quad+(k-N+2)(m_{N-2}-m_{k-N+2})+(2N-2-k)m_{k-N+1}+(2N-3-k)m_{k-N+2}+(4N-2k-5)H_{N,a,b}.
\end{align*}

Similarly,
\begin{align*}
    \sum_{j=k-N+1}^{N-2} a_j&=\sum_{j=k-N+1}^{N-2} (H_{N,a,b}+m_j)=(2N-k-2)H_{N,a,b}+\sum_{j=k-N+1}^{N-2}m_j
\end{align*}

Hence, for $N-1\leq k\leq 2N-5$
\begin{align*}
p_k&=b_{k-1}+b_k-\sum_{j=k-N+1}^{N-2} a_j
\\&=\sum_{j=k-N+3}^{N-2}(m_{k-j}-2m_j+m_{k-j+1})j 
\\&\quad+(k-N+2)(m_{N-2}-m_{k-N+2})+(2N-2-k)m_{k-N+1}+(2N-3-k)m_{k-N+2}+(4N-2k-5)H_{N,a,b}
\\&\qquad-(2N-k-2)H_{N,a,b}-\sum_{j=k-N+1}^{N-2} m_j
\\&=\sum_{j=k-N+3}^{N-2}(m_{k-j}-2m_j+m_{k-j+1})j 
\\&\quad+(k-N+2)(m_{N-2}-m_{k-N+2})+(2N-2-k)m_{k-N+1}+(2N-3-k)m_{k-N+2}
\\&\qquad+(2N-k-3)H_{N,a,b}-\sum_{j=k-N+1}^{N-2} m_j.
\end{align*}

For $N-1\leq k\leq 2N-5$, let $k:=2N-4-\hat{k}$, then $1\leq \hat{k}\leq N-3$ and
\begin{align*}
p_{2N-4-\hat{k}}&=\sum_{j=N-\hat{k}-1}^{N-2}(m_{2N-4-\hat{k}-j}-2m_j+m_{2N-3-\hat{k}-j})j+(N-\hat{k}-2)(m_{N-2}-m_{N-\hat{k}-2})
\\&\qquad +(\hat{k}+2)m_{N-\hat{k}-3}+(\hat{k}+1)m_{N-\hat{k}-2}+(\hat{k}+1)H_{N,a,b}-\sum_{j=N-\hat{k}-3}^{N-2}m_j
\end{align*}

Now we will show that $p_{2N-4-\hat{k}}+p_{\hat{k}}=0$, for all $1\leq \hat{k}\leq N-3$. We have
\begin{align}
p_{2N-4-\hat{k}}+p_{\hat{k}}&=\sum_{j=N-\hat{k}-1}^{N-2}(m_{2N-4-\hat{k}-j}-2m_j+m_{2N-3-\hat{k}-j})j+(N-\hat{k}-2)(m_{N-2}-m_{N-\hat{k}-2})\notag
\\&\qquad +(\hat{k}+2)m_{N-\hat{k}-3}+(\hat{k}+1)m_{N-\hat{k}-2}+(\hat{k}+1)H_{N,a,b}-\sum_{j=N-\hat{k}-3}^{N-2}m_j\notag
\\& \qquad +\sum_{j=1}^{\hat{k}}
(m_{\hat{k}-j+1} + m_{\hat{k}-j} - 2m_j)j +(m_0-m_{\hat{k}+1})(\hat{k}+1)-(k+1)H_{N,a,b}-\sum_{j=0}^{\hat{k}}m_j\notag
\\&=\sum_{j=N-\hat{k}-1}^{N-2}(m_{2N-4-\hat{k}-j}-2m_j+m_{2N-3-\hat{k}-j})j+(N-\hat{k}-2)m_{N-2}\notag
\\&\qquad-(N-2\hat{k}-3)m_{N-\hat{k}-2} +(\hat{k}+2)m_{N-\hat{k}-3}-\sum_{j=N-\hat{k}-3}^{N-2}m_j\notag
\\& \qquad + \sum_{j=1}^{\hat{k}}
(m_{\hat{k}-j+1} + m_{\hat{k}-j} - 2m_j)j +(m_0-m_{\hat{k}+1})(\hat{k}+1)-\sum_{j=0}^{\hat{k}}m_j\label{eq:sum1}
\end{align}

To proceed further, we need to simplify the two sums
$$
S_1=\sum_{j=1}^{\hat{k}}
(m_{\hat{k}-j+1} + m_{\hat{k}-j} - 2m_j)j\quad \text{and}\quad S_2=\sum_{j=N-\hat{k}-1}^{N-2}(m_{2N-4-\hat{k}-j}-2m_j+m_{2N-3-\hat{k}-j})j
$$

using discrete integration by parts techniques.

For the first sum in $S_1$, by changing of index $j$ to $\hat{k}-j+1$, we get
$$
\sum_{j=1}^{\hat{k}}
m_{\hat{k}-j+1}\, j=\sum_{j=1}^{\hat{k}}
m_{j}\, (\hat{k}-j+1).
$$

Similarly, by changing of index $j$ to $\hat{k}-j$ we can rewrite the second sum in $S_1$ as
$$
\sum_{j=1}^{\hat{k}}
m_{\hat{k}-j}\, j=\sum_{j=0}^{\hat{k}-1}
m_{j}\, (\hat{k}-j)=\sum_{j=1}^{\hat{k}}
m_{j}\, (\hat{k}-j)+\hat{k} m_0.
$$

Hence, we have
\begin{align}
S_1&=\sum_{j=1}^{\hat{k}}
(m_{\hat{k}-j+1} + m_{\hat{k}-j} - 2m_j)j\notag
\\&=\sum_{j=1}^{\hat{k}}
m_{j}\, (\hat{k}-j+1)+\sum_{j=1}^{\hat{k}}
m_{j}\, (\hat{k}-j)- 2\sum_{j=1}^{\hat{k}}
m_{j}\,j+\hat{k} m_0\notag
\\&= \sum_{j=1}^{\hat{k}} (2(\hat{k}-2j)+1)m_j+ \hat{k}m_0. \label{eq: S1}
\end{align}

By proceeding similarly, we can transform $S_2$ as
\begin{align}
S_2&=\sum_{j=N-\hat{k}-1}^{N-2}(m_{2N-4-\hat{k}-j}-2m_j+m_{2N-3-\hat{k}-j})j \notag
\\&= \sum_{j=N-\hat{k}-2}^{N-2} m_j(4N-2\hat{k}-4j-7) - (N-\hat{k}-2)m_{N-2} - (N-1)m_{N-\hat{k}-2}+2(N-\hat{k}-2)m_{N-\hat{k}-2}\notag
\\&= \sum_{j=N-\hat{k}-2}^{N-2} m_j(4N-2\hat{k}-4j-7) - (N-\hat{k}-2)m_{N-2}+(N-2\hat{k}-3)m_{N-\hat{k}-2}.\label{eq: S2}
\end{align}

Substituting \eqref{eq: S1} and \eqref{eq: S2} back into \eqref{eq:sum1} we have
\begin{align*}
p_{2N-4-\hat{k}}+p_{\hat{k}}&=\sum_{j=N-\hat{k}-2}^{N-2} m_j(4N-2\hat{k}-4j-7) - (N-\hat{k}-2)m_{N-2} +(N-2\hat{k}-3)m_{N-\hat{k}-2}+ (N-\hat{k}-2)m_{N-\hat{k}-2}
\\& \qquad-(N-2\hat{k}-3)m_{N-\hat{k}-2} +(\hat{k}+2)m_{N-\hat{k}-3}-\sum_{j=N-\hat{k}-3}^{N-2}m_j
\\& \qquad+\sum_{j=1}^{\hat{k}} (2(\hat{k}-2j)+1)m_j+ \hat{k}m_0  +(m_0-m_{\hat{k}+1})(\hat{k}+1)-\sum_{j=0}^{\hat{k}}m_j.
\end{align*}

It follows immediately from the formula of $m_j$ that when $a=b$.
$$
\sum_{j=0}^{N-2}m_j=H_{N,a,b}, 
$$

and $m_{j}=m_{N-2-j}$ for $j=0,\ldots, N-2$.

Therefore, we have
\begin{align*}
\sum_{j=N-\hat{k}-2}^{N-2} m_j(4N-2\hat{k}-4j-7)&=\sum_{j=N-\hat{k}-2}^{N-2} m_{N-2-j}(4N-2\hat{k}-4j-7)=\sum_{i=0}^{\hat{k}} m_i (-2\hat{k}+4i+1)
\\&=\sum_{j=1}^{\hat{k}} m_j (-2\hat{k}+4j+1)+m_0(-2\hat{k}+1).
\end{align*}

So we have
\begin{align*}
p_{2N-4-\hat{k}}+p_{\hat{k}}&= 2 \sum_{j=1}^{\hat{k}}m_j+(\hat{k}+2)m_{N-\hat{k}-3}-\sum_{j=N-\hat{k}-3}^{N-2}m_j+m_0(-2\hat{k}+1)+\hat{k} m_0+(m_0-m_{\hat{k}+1})(\hat{k}+1)-\sum_{j=0}^{\hat{k}}m_j
\\&=\sum_{j=0}^{\hat{k}}m_j+(\hat{k}+2)m_{N-\hat{k}-3}-(\hat{k}+1)m_{\hat{k}+1}-\sum_{j=N-\hat{k}-3}^{N-2}m_j
\\&=\sum_{j=0}^{\hat{k}}m_j+(\hat{k}+2)m_{N-\hat{k}-3}-(\hat{k}+1)m_{\hat{k}+1}-\sum_{j=0}^{\hat{k}}m_j-m_{\hat{k}+1}
\\&=0,
\end{align*}
where in the last equality, we have used the fact that $m_{N-\hat{k}-3}=m_{\hat{k}+1}$.

Next we show that $p_k<0$ for $k=1,\ldots, N-3$. Substituting \eqref{eq: S1} back into \eqref{eq:pk} we get for $k=1,\ldots, N-3$
\begin{align}
p_k&=\sum_{j=1}^{k} (2(k-2j)+1)m_j+ km_0.  +(m_0-m_{k+1})(k+1)-(k+1)H_{N,a,b}-\sum_{j=0}^{k}m_j\notag
\\&=\sum_{j=0}^{k} 2(k-2j)m_j  -(k+1)m_{k+1}-(k+1)H_{N,a,b}\notag
\\&=\sum_{j=0}^{\lfloor k/2\rfloor} 2\,k\, m_j-\sum_{j=0}^{\lfloor k/2\rfloor} 4jm_j + \sum_{j=\lfloor k/2\rfloor+1}^{k} 2(k-2j)m_j  -(k+1)m_{k+1}-(k+1)H_{N,a,b}\label{eq:pk2}
\end{align}

In the above expression, only the first term is positive, the other ones are negative. We will show that the first term is dominated by the last term. We recall that
$H_{N,a,b}=\sum_{j=0}^{N-2}m_j$ and for $a=b$, then $m_j=m_{N-2-j}$. Hence
$$
H_{N,a,b}=\sum_{j=0}^{\lfloor (N-2)/2 \rfloor} m_j+\sum_{\lfloor (N-2)/2 \rfloor+1}^{N-2} m_j=2\sum_{j=0}^{\lfloor (N-2)/2 \rfloor} m_j
$$

Substituting this back to \eqref{eq:pk2} we obtain
\begin{align*}
p_k&=\sum_{j=0}^{\lfloor k/2\rfloor} 2\,k\, m_j-\sum_{j=0}^{\lfloor k/2\rfloor} 4jm_j + \sum_{j=\lfloor k/2\rfloor+1}^{k} 2(k-2j)m_j  -(k+1)m_{k+1}-2(k+1)\sum_{j=0}^{\lfloor (N-2)/2 \rfloor} m_j
\end{align*}

Since $k\leq N-3$ we have
$$
\sum_{j=0}^{\lfloor k/2\rfloor} 2\,k\, m_j<2(k+1)\sum_{j=0}^{\lfloor (N-2)/2 \rfloor} m_j.
$$

Therefore, $p_k<0$ for all $k=0,\ldots, N-3$. In conclusion,  we have for $0\leq k\leq N-2$
$$
p_k=-p_{2N-4-k}<0.
$$

As a consequence, when $a=b$, the sequence of coefficients of $P$ has exactly one change of sign. Thus according to Descartes' rule of sign, $P$ has exactly one positive root. Because $p_k=-p_{2N-4-k}$ for all $k=0,\ldots, N-2$, the root is equal to $1$. This finishes the proof of this proposition.
\end{proof}
\subsection{Proof of Proposition \ref{prop: M}}
\label{sec:proof2}
In this section, we provide a detailed proof of Proposition \ref{prop: M}.
\begin{proof}
Recall that
$$
Q(u)=(1+u)\Big(\sum_{j=0}^{N-2}a_j u^j\Big)\Big(\sum_{j=0}^{N-1} u^{j}\Big)
$$
where
$$
a_j= H_{N,a,b}+m_j.
$$

Thus $Q$ is a polynomial of degree $2N-2$ and the leading coefficient is $q_{2N-2}=a_{N-2}$, which is
$$
a_{N-2}=H_{N,a,b}+m_{N-2}=H_{N,a,b}+\frac{ \min{(\frac{N-1}{a}, \frac{1}{b})}}{N-1}>0.
$$

Hence $Q'(u)$ is a polynomial of degree $2N-3$ whose leading coefficient is $(2N-2)a_{N-2}$.

According to Proposition \ref{prop: coeffs} $P$ is a polynomial of degree $2N-4$ whose leading coefficient is
$$
p_{2N-4}=a_{N-3}=H_{N,a,b} + \frac{\min(\frac{N-2}{a},\frac{2}{b})}{2(N-2)}.
$$

Hence $P'$ is a polynomial of degree $2N-5$ whose leading coefficient is $(2N-4)a_{N-3}$.

It follows that $M$ is a polynomial of degree 
$$
\max \{(2N-3)+(2N-4)+1, (2N-2)+(2N-4), (2N-2)+(2N-5)+1, 1+2(2N-4)\}=4N-6.
$$

with the leading coefficient
$$
m_{4N-6}=(2N-2)a_{N-2} a_{N-3}-a_{N-2}a_{N-3}-(2N-4)a_{N-2}a_{N-3}=a_{N-2}a_{N-3}>0.
$$

We write
$$
M(u)=\sum_{i=0}^{4N-6} m_i u^i.
$$

Next we prove that the coefficients of $M$ are symmetric, that is $m_{i}=m_{4N-6-i}$ for all $i=0,\ldots, 4N-6$. In fact, we observe that
$$
u^{4N-6} M(1/u)=u^{4N-6}\sum_{i=0}^{4N-6} m_i (1/u)^i=\sum_{i=0}^{4N-6} m_i u^{4N-6-i}=\sum_{i=0}^{4N-6} m_{4N-6-i} u^{i}
$$

Thus to prove that the coefficients of $M$ are symmetric is equivalent to showing that
\begin{equation}
\label{eq: needprooof}
u^{4N-6}M(1/u)=M(u).
\end{equation}

We now prove this equality. We have
\begin{align}
u^{2N-2}Q(1/u)&=u^{2N-2}(1+1/u)\Big(\sum_{j=0}^{N-2}a_j (1/u)^j\Big)\Big(\sum_{j=0}^{N-1} (1/u)^{j}\Big)\notag
\\&=(1+u)\Big(\sum_{j=0}^{N-2}a_{j} u^{N-2-j}\Big)\Big(\sum_{j=0}^{N-1} u^{N-1-j}\Big)\notag
\\&=(1+u)\Big(\sum_{j=0}^{N-2}a_{N-2-j} u^{j}\Big)\Big(\sum_{j=0}^{N-1} u^j\Big)\notag
\\&=(1+u)\Big(\sum_{j=0}^{N-2}a_{j} u^{j}\Big)\Big(\sum_{j=0}^{N-1} u^j\Big)
\\&= Q(u), \label{eq: Q1u}
\end{align}
where we have used the fact that $a_j=a_{N-2-j}$ for all $j=0,\ldots, N-2$. It follows that
\begin{equation}
Q'(u)=\Big[u^{2N-2}Q(1/u)\Big]'=(2N-2) u^{2N-3} Q(1/u)-u^{2N-4}Q'(1/u).\label{eq: Qder}
\end{equation}

Therefore
\begin{equation}
Q'(1/u)=\frac{1}{u^{2N-4}}\Big[(2N-2) u^{2N-3} Q(1/u)-Q'(u)\Big].
\label{eq: Qder1u}
\end{equation}

We now do similar computations for $P$. We have
\begin{align}
u^{2N-4}P(1/u)&=u^{2N-4}\sum_{i=0}^{2N-4} p_i (1/u)^i\notag
\\&=\sum_{i=0}^{2N-4} p_i u^{2N-4-i}\notag
\\&=\sum_{i=0}^{2N-4} p_{2N-4-i} u^{i}\notag
\\&=-\sum_{i=0}^{2N-4} p_{i} u^{i}\notag
\\&=-P(u),\label{eq: P1u}
\end{align}
where we have used the fact that $p_{2N-4-i}=-p_i$ for all $i=0,\ldots, 2N-4$. It also follows that
\begin{equation}
P'(u)=-\Big[u^{2N-4}P(1/u)\Big]'=-\Big[(2N-4)u^{2N-5}P(1/u)-u^{2N-6}P'(1/u)\Big].
\label{eq: Pder}
\end{equation}

Thus
\begin{equation}
P'(1/u)=\frac{1}{u^{2N-6}}\Big[P'(u)+(2N-4)u^{2N-5}P(1/u)\Big].
\label{eq:Pder1u}
\end{equation}

Therefore, from \eqref{eq: Q1u} to \eqref{eq:Pder1u} we have
\begin{align*}
M(1/u)&= (1/u) Q'(1/u)P(1/u)-Q(1/u)(P(1/u)+(1/u)P'(1/u))-(1/u)P(1/u)^2
\\&=(1/u)\frac{1}{u^{2N-4}}\Big[(2N-2) u^{2N-3} Q(1/u)-Q'(u)\Big]\Big[-1/u^{2N-4} P(u)\Big]
\\&\qquad-\Big[1/u^{2N-2}Q(u)\Big]\bigg(-1/u^{2N-4}P(u)+(1/u)\frac{1}{u^{2N-6}}\Big[P'(u)+(2N-4)u^{2N-5}P(1/u)\Big]\bigg)
\\&\qquad-(1/u)\Big[-1/u^{2N-4}P(u)\Big]^2
\\&=\frac{1}{u^{4N-7}} Q'(u)P(u)+\frac{1}{u^{4N-6}}Q(u)P(u)-\frac{1}{u^{4N-7}}Q(u)P'(u)-\frac{1}{u^{4N-7}} P(u)^2
\\&\qquad-\frac{1}{u^{2N-4}}(2N-2)Q(1/u)P(u)-\frac{1}{u^{2N-2}}(2N-4)Q(u)P(1/u)
\end{align*}

Hence
\begin{align}
u^{4N-6} M(1/u)&= u Q'(u)P(u)-uQ(u)P'(u)-u P(u)^2\notag
    \\& + Q(u)P(u)-u^{2N-2}(2N-2)Q(1/u)P(u)-u^{2N-4}(2N-4)Q(u)P(1/u). \label{eq: M1u}    
\end{align}

Using \eqref{eq: Q1u} and \eqref{eq: P1u}, we can simplify further the last line in the above expression as follows
\begin{align*}
&Q(u)P(u)-u^{2N-2}(2N-2)Q(1/u)P(u)-(2N-4)Q(u)P(1/u)
\\&\qquad=Q(u)P(u)-(2N-2) Q(u)P(u)+(2N-4)Q(u)P(u)
\\&\qquad=-Q(u)P(u). 
\end{align*}

Substituting this back into \eqref{eq: M1u} we obtain
\begin{equation*}
 u^{4N-6} M(1/u)= u Q'(u)P(u)-uQ(u)P'(u)-u P(u)^2-Q(u)P(u)=M(u),
\end{equation*}
which is the desired equality \eqref{eq: needprooof}. Hence the coefficients of $M$ are symmetric.

Since $m_0=m_{4N-6}>0$, it implies that
\begin{equation}
\label{eq: M0}
M(0)>0, \quad M(+\infty)=\lim_{u\rightarrow +\infty} M(u)>0.
\end{equation}

Next we show that $M$ has at least to positive root by showing that $M(1)<0$.

From \eqref{eq: M}, since $P(1)=0$ we have
$$
M(1)=-Q(1)P'(1).
$$

Clearly $Q(1)= 2N \sum_{j=0}^{N-2}a_j >0$. Since $P(u) = \sum_{k=0}^{2N-4} p_k u^k$, we have $P'(u) = \sum_{k=0}^{2N-4} k p_k u^{k-1}$. Hence $P'(1)=\sum_{k=0}^{2N-4} k p_k$. According to Proposition \ref{prop: coeffs} $p_k=-p_{2N-4-k}$ for $k=0,\ldots, N-3$ and $p_{N-2}=0$. Thus we can rewrite $P'(1)$ as follows.
\begin{align*}
    P'(1)&= \sum_{k=0}^{2N-4} k p_k 
    \\&= \sum_{k=0}^{N-3} k p_k  + \sum_{\hat{k}=N-1}^{2N-4} \hat{k} p_{\hat{k}}
    \\&= \sum_{k=0}^{N-3} k p_k + \sum_{k=0}^{N-3} (2N-4+k) p_{2N-4-k}
    \\&= \sum_{k=0}^{N-3} (2N-4+k-k) p_{2N-4-k}
    \\&= (2N-4) \sum_{k=0}^{N-3} p_{2N-4-k},
\end{align*} 

Note that in the above computations, to go from the second equality to the third one, we have used a change of variable $\hat{k}= 2N-4-k$. Since $p_{2N-4-k}>0$ for all $k=0,\ldots, N-3$, $P'(1)>0$ for all $N\geq 3$.  Thus, as $M(1)=-Q(1)P'(1)$ and $Q(1)>0, P'(1)>0$, we obtain that $M(1)<0$. This, together with \eqref{eq: M0} and since $M$ is a polynomial, we deduce, by the Intermediate Value Theorem, that $M(u)$ has at least 2 roots, one in the interval $(0,1)$ and another in the interval $(1,+\infty)$. 
\end{proof}
\subsection{Proof of Lemma \ref{lem: bound of HNab}}
\label{sec: proof3}
\begin{proof}
We have
$$
\min\Big(\frac{j}{\max(a,b)},\frac{N-j}{\max(a,b)}\Big) \leq \min\Big(\frac{j}{a},\frac{N-j}{b}\Big) \leq \min\Big(\frac{j}{\min(a,b)},\frac{N-j}{\min(a,b)}\Big).
$$

We observe that
$$
\min\Big(\frac{j}{\max(a,b)},\frac{N-j}{\max(a,b)}\Big)=\frac{1}{\max(a,b)}\min(j,N-j),
$$

and

$$
\min\Big(\frac{j}{\min(a,b)},\frac{N-j}{\min(a,b)}\Big)=\frac{1}{\min(a,b)}\min(j,N-j).
$$

Thus we have
\begin{equation}
\label{eq: HNab vs HN11}
\frac{1}{\max(a,b)}\sum\limits_{j=1}^{N-1} \frac{\min(j,N-j)}{j(N-j)} \leq H_{N,a,b} \leq \frac{1}{\min(a,b)}\sum\limits_{j=1}^{N-1} \frac{\min(j,N-j)}{j(N-j)}.
\end{equation}

Let $H_{N,1,1}=\sum\limits_{j=1}^{N-1} \frac{\min(j,N-j)}{j(N-j)}$, which is exactly $H_{N,a,b}$ when $a=b=1$. Then we have
\begin{align*}
    H_{N,1,1}=\sum\limits_{j=1}^{N-1} \frac{\min(j,N-j)}{j(N-j)}
    =\sum\limits_{j=1}^{\lfloor{\frac{N}{2}}\rfloor} \frac{1}{N-j} + \sum\limits_{\lfloor{\frac{N}{2}}\rfloor+1}^{N-1} \frac{1}{j}= 2 \sum\limits_{\lfloor{\frac{N}{2}}\rfloor+1}^{N-1} \frac{1}{j}=2(H_N-H_{\lfloor{\frac{N}{2}}\rfloor}).
\end{align*}

where 

$$
H_N=\sum_{j=1}^{N-1}\frac{1}{j}
$$

is the harmonic number. This important number plays a central role in number theory. Its appearance in our analysis, as well as in \cite{duong2021cost}, is rather interesting. Using the above relationship and the following well-known estimates for the harmonic number
\begin{equation}
\label{eq: HN estimates}    
\ln N +\gamma +\frac{1}{2N+1}\leq H_N\leq \ln N +\gamma +\frac{1}{2N-1},
\end{equation}

where $\gamma=0.5772156649$ is the celebrated Euler–Mascheroni constant, we obtain the following lower and upper bounds for $H_{N,1,1}$:

\begin{align*}
&H_{N,1,1}=2(H_N-H_{\lfloor{\frac{N}{2}}\rfloor})\leq 2\Big(\ln N +\gamma +\frac{1}{2N-1}-\ln \lfloor \frac{N}{2} \rfloor\Big) -\gamma - \frac{1}{2\lfloor \frac{N}{2}) \rfloor +1}= 2\Big(\ln 2 + \frac{1}{2N-1}-\frac{1}{N+1}\Big),
\\&H_{N,1,1}=2\Big(H_N-H_{\lfloor{\frac{N}{2}}\rfloor}\Big)
 \geq 2\Big(\ln N +\gamma +\frac{1}{2N+1}-\ln \lfloor \frac{N}{2} \rfloor\Big) -\gamma - \frac{1}{2\lfloor \frac{N}{2}) \rfloor -1}= 2\Big(\ln 2 + \frac{1}{2N+1}-\frac{1}{N-1}\Big)
\end{align*}

Substituting these estimates back to \eqref{eq: HNab vs HN11} we obtain the desired lower and upper bounds for $H_{N,a,b}$ and complete the proof of this lemma.
\end{proof}
\subsection{Proof of Lemma \ref{lem: m and M}}
\label{sec: proof4}
\begin{proof}
We have, for $i=0,\ldots,N-2$
\begin{align*}
  m=\min_{i}\frac{\min(\frac{i+1}{a},\frac{N-i-1}{b})}{(i+1)(N-i-1)}\geq \min_{i}\frac{\min(\frac{i+1}{\max(a,b)},\frac{N-i-1}{\max(a,b)})}{(i+1)(N-i-1)}\geq \min_{i}\frac{\frac{\min(i+1,N-i-1)}{\max(a,b)}}{(i+1)(N-i-1)},
\end{align*} 

We have
$$
\min\Big\{\frac{\min(1,N-1)}{N-1},\ldots,\frac{\min(N-1,1)}{N-1}\Big\}
= \min\Big\{\frac{1}{N-1},\ldots,\frac{1}{N-1}\Big\} = \frac{1}{N-1}.
$$

Therefore, we obtain
\begin{align*}
    m\geq \frac{1}{\max(a,b)(N-1)}.
\end{align*}

Similarly for $M$
\begin{align*}
  M=\max_{i}\frac{\min(\frac{i+1}{a},\frac{N-i-1}{b})}{(i+1)(N-i-1)}\leq \max_{i}\frac{\min(\frac{i+1}{\min(a,b)},\frac{N-i-1}{\min(a,b)})}{(i+1)(N-i-1)}\leq \max_{i}\frac{\frac{\min(i+1,N-i-1)}{\min(a,b)}}{(i+1)(N-i-1)},
\end{align*} 
for $i=0,\ldots,N-2$.

We have
$$
\max\Big\{\frac{\min(1,N-1)}{N-1},\ldots,\frac{\min(N-1,1)}{N-1}\Big\}
= \max\Big\{\frac{1}{N-1},\ldots,\frac{1}{\lfloor \frac{(N-1)}{2} \rfloor} \Big\} = \frac{1}{\lfloor \frac{(N-1)}{2} \rfloor}.
$$

Thus
\begin{align*}
    M\leq \frac{1}{\min(a,b)\lfloor \frac{(N-1)}{2} \rfloor}.
\end{align*}

Substituting the lower bound for $m$ and upper bound for $M$ into \eqref{eq: estimates1} we obtain the desired estimates for $E_{mix}$. 
\end{proof}
\section*{Acknowledgement} The research of MHD was supported by EPSRC Grants EP/W008041/1 and EP/V038516/1. We would like to thank the anonymous referees for their instructive comments and suggestions, which were very helpful for us to improve the manuscript.
\section*{Data Availability Statement}
Data sharing is not applicable to this article as no datasets were generated or analysed during the current study.
\bibliographystyle{plain}
\bibliography{references}
\end{document}